\documentclass[acmsmall,screen,nonacm]{acmart}

\AtBeginDocument{%
  \providecommand\BibTeX{{%
    \normalfont B\kern-0.5em{\scshape i\kern-0.25em b}\kern-0.8em\TeX}}}

\setcopyright{acmcopyright}
\copyrightyear{2021}
\acmYear{2021}
\acmDOI{10.1145/1122445.1122456}

\acmConference[Woodstock '18]{Woodstock '18: ACM Symposium on Neural
  Gaze Detection}{June 03--05, 2018}{Woodstock, NY}
\acmBooktitle{Woodstock '18: ACM Symposium on Neural Gaze Detection,
  June 03--05, 2018, Woodstock, NY}
\acmPrice{15.00}
\acmISBN{978-1-4503-XXXX-X/18/06}

\usepackage[ruled,vlined]{algorithm2e}
\usepackage{subfigure}
\usepackage{wrapfig}
\usepackage{mdframed}

\newtheorem{theorem}{Theorem}
\newtheorem{lemma}{Lemma}

\newtheorem{corollary}{Corollary}

\usepackage{thmtools}

\declaretheorem[
  shaded={rulecolor=black, rulewidth=0.5pt, bgcolor=white},
  name=Theorem,
]{thmboxed}

\newtheorem*{T2}{Theorem~\ref{thm:analytical-calculation-of-success-rate}}
\newtheorem*{L2}{Lemma~\ref{lemma:channel-with-buffers}}

\begin{document}

% \begin{center}
% \Large{Optimizing the Throughput of Payment Channel Networks via State-Dependent Payment Processing}
% \end{center}
% \vspace{22pt}

% \begin{center}
% Nikolaos Papadis, Leandros Tassiulas\\
% Yale University\\
% Draft date: \today
% \end{center}

\title[]{State-Dependent Processing in Payment Channel Networks for Throughput Optimization}
% Throughput-Optimal Transaction Scheduling in Blockchain-Based Payment Channel Networks

% \affiliation{%

\author{Nikolaos Papadis}
\affiliation{%
    \institution{Electrical Engineering \& Institute for Network Science, Yale University}
  \city{New Haven}
  \state{Connecticut}
  \country{USA}}
\email{nikolaos.papadis@yale.edu}

\author{Leandros Tassiulas}
\affiliation{%
  \institution{Electrical Engineering \& Institute for Network Science, Yale University}
  \city{New Haven}
  \state{Connecticut}
  \country{USA}}
\email{leandros.tassiulas@yale.edu}

\renewcommand{\shortauthors}{}

\begin{abstract}
Payment channel networks (PCNs) have emerged as a scalability solution for blockchains built on the concept of a payment channel: a setting that allows two nodes to safely transact between themselves in high frequencies based on pre-committed peer-to-peer balances. 
Transaction requests in these networks may be declined because of unavailability of funds due to temporary uneven distribution of the channel balances. 
In this paper, we investigate how to alleviate unnecessary payment blockage via proper prioritization of the transaction execution order.
Specifically, we consider the scheduling problem in PCNs: as transactions continuously arrive on both sides of a channel, nodes need to decide which ones to process and when in order to maximize their objective, which in our case is the channel throughput.
We introduce a stochastic model to capture the dynamics of a payment channel under random arrivals, and propose that channels can hold incoming transactions in buffers up to some deadline in order to enable more elaborate processing decisions.
We describe a policy that maximizes the channel success rate/throughput for uniform transaction requests of fixed amounts, both in the presence and absence of buffering capabilities, and formally prove its optimality.
We also develop a discrete event simulator of a payment channel, and evaluate different heuristic scheduling policies in the more general heterogeneous amounts case, with the results showing superiority of the heuristic extension of our policy in this case as well.
Our work opens the way for more formal research on improving PCN performance via joint consideration of routing and scheduling decisions.

\end{abstract}

\maketitle

% \begin{figure}[h]
%   \centering
% %   \includegraphics[width=\linewidth]{sample-franklin}
%   \caption{1907 Franklin Model D roadster. Photograph by Harris \&
%     Ewing, Inc. [Public domain], via Wikimedia
%     Commons. (\url{https://goo.gl/VLCRBB}).}
%   \Description{A woman and a girl in white dresses sit in an open car.}
% \end{figure}

\section{Introduction}

Blockchain technology enables trusted collaboration between untrusted parties that want to reach consensus in a distributed setting. 
This is achieved with the help of a distributed ledger, which is maintained by all interested nodes in the network and functions as the source of truth.
The original application of blockchain, Bitcoin \cite{Nakamoto2008}, as well as many subsequent ones, focus on the problem of distributed consensus on a set of financial transactions and the order in which they were executed.
Agreement on the above provides a way for everyone to be able to prove that they own the amount they claim to own, without a central entity such as a bank, which is a trusted institution charged with this role in the traditional economic activity.
Transactions are organized in blocks and blocks are chained to form the ledger, or the \textit{blockchain}.
In order for a single node to be able to amend history to its benefit, significant power in the network is required.
In Proof of Work blockchains (including Bitcoin) for example, which rely on nodes expending computation on solving a hard hash puzzle to include their block in the chain, the attacking node should own a certain fraction of the computational power, while in Proof of Stake blockchains, which rely on nodes staking their wealth in order to publish new blocks, the attacker should own a certain fraction of the network's stake.
Accountability and transparency are thus guaranteed as long as each node's share in the network power is limited.

Despite their success with solving distributed consensus, a major pain point of blockchains is their scalability \cite{Cromanetal2016, Papadis2018, prism}.
Compared to a centralized system, where everyone communicates with a single entity functioning as the source of truth, decentralizing this operation and assigning this role to the entire network introduces significant overheads in communication and in complexity.
The frequently cited figures for the transactions per second (throughput) achieved by the two most prominent cryptocurrencies, 3-7 for Bitcoin and about double that for Ethereum, are a good indication of the scalability problem, especially as centralized counterparts such as PayPal or Visa achieve throughput of thousands of transactions per second.
Therefore, for blockchain to be a long-term viable payment solution, this scalability barrier has to be overcome.

A promising development in the scalability front is brought by the introduction of payment channel networks (PCNs).
PCNs are a ``layer-2'' solution based on the idea that the majority of transactions are only communicated to the interested parties instead of the entire global network, and the global network is only consulted in case of disputes. 
The main building block of a PCN is the concept of a payment channel: two entities from layer-1 (the blockchain network itself) that want to transact frequently between themselves and do not need nor want the entire network confirming and knowing, can form a payment channel via a smart contract recorded on the blockchain and validated by the entire network. 
After the channel is created, the nodes can transact privately and orders of magnitude faster than done via the main layer-1 network. 
Payment channels form a network themselves, the PCN, in which multihop payments are possible, and intermediate nodes relaying payments make profit from collected fees. The most prominent PCN as of now are the Lightning Network \cite{Poon2016} and the Raiden Network \cite{RaidenNetwork}.

Sending payments via the network formed by the channels requires appropriate payment routing, scheduling, and congestion control, to guarantee sufficient success rates and throughput. 
A multi-hop transaction might fail if it encounters a channel with insufficient balance to process it on its path.
Several routing approaches have been proposed for proper path selection \cite{Papadis2020}, including source routing \cite{Poon2016}, max-flow-based approaches \cite{Sivaraman2020, rohrer2017, Yu2018a, celer, Wang2019, Varma2019}, beacon-based routing with proactive aggregation of information \cite{Prihodko2016}, landmark-based routing \cite{Silentwhispers}, embedding-based routing \cite{Speedymurmurs}, distance-vector routing \cite{Hoenisch2018}, and ant routing \cite{Grunspan2018b}. 
Scheduling and congestion control have received little attention, with the notable exception of \cite{Sivaraman2020}.
Most of these schemes employ some heuristic rules and lack formal optimality guarantees.

In this work, we study the transaction scheduling problem is PCNs from a formal point of view.
As transactions continuously arrive at the two sides of each channel, the nodes have to make scheduling decisions: which transactions to process, and when.
We introduce a stochastic model for the channel's operation and derive an optimal policy that allows the channel to operate at the maximum possible throughput, which is beneficial both for the nodes relaying others' payment to collect fees, and for the network overall.
In addition, we advocate for a modification in how transactions are handled by nodes: we introduce pending transaction buffers (queues) at the nodes, and allow the transactions to specify a deadline up to which their sender is willing to wait in order to increase their success probability.
The rationale behind this modification is that an initially infeasible transaction, in the extra time it is given in the buffer compared to being rejected immediately, might become feasible thanks to the updates in the channel balances from transactions executed from the opposite side.
Thus, more elaborate state-dependent scheduling policies become possible, making decisions based not only on the instantaneous balances, but also on the buffer contents (the pending transactions, each with their direction, amount and remaining time to expiration).
In this general setting, we are the first to analytically describe a throughput-maximizing scheduling policy for a payment channel and prove its optimality among all dynamic policies.
Our theoretical results are complemented by experiments in a payment channel simulator we implemented, and on which we test various policies and compare their performance.
% Our results are directly applicable to PCNs that form a complete graph, and we discuss extensions and difficulties arising therein when one examines different types of graphs.

In summary, our contributions and insights are the following:
\begin{itemize}

    \item We develop a stochastic model that captures the dynamics of a payment channel in an environment with random transaction arrivals from both sides.

    \item We propose the idea of transaction deadlines and buffering in order to give nodes more freedom in their scheduling decisions, and formulate the scheduling problem in our stochastic model, for a channel both without and with buffering capabilities.

    \item We describe policies that optimize the throughput, the success rate and the blockage when transaction amounts are fixed, and present the optimality proofs for a channel both without and with buffering capabilities. 
    We also introduce two families of heuristic policies for the arbitrary amounts case.

    \item We develop a realistic payment channel simulator that accounts for the simultaneity of payments and implements the node buffering capabilities. 
    We use the simulator to evaluate the different scheduling policies in both the fixed and varying transaction amounts cases.

    \item We discuss the necessity of a joint approach to the fundamental problems of routing and scheduling, using either formal stochastic modeling techniques, or learning-based techniques that leverage the network's operation data.
\end{itemize}

In summary, our paper is the first to formally treat the optimal scheduling problem in a PCN with buffering capabilities.

The remainder of the paper is organized as follows. 
In section \ref{sec:background} we provide an introduction to the operation of payment channels and introduce the idea of transaction buffers. 
In section \ref{sec:problem-formulation} we describe our stochastic model of a payment channel, and in section \ref{sec:policy-section} we present the throughput-optimal scheduling policies, whose optimality we subsequently prove.
In section \ref{sec:heuristic-policies} we present heuristic policies for the more general arbitrary amounts case, and in section \ref{sec:evaluation-results} we describe the experimental setup and the simulator used for the evaluation, and present the results of several experiments we conducted.
In section \ref{sec:extensions} we discuss extensions and generalizations of this work to arbitrary network structures, and in section \ref{sec:related-work} we look into related work.
Finally, section \ref{sec:conclusion} concludes the paper.

\section{Background}
\label{sec:background}

\paragraph{Payment channel operation}

\begin{wrapfigure}{R}{0.5\textwidth}
% \begin{figure}[h]
    \centering
    \includegraphics[width=0.5\textwidth]{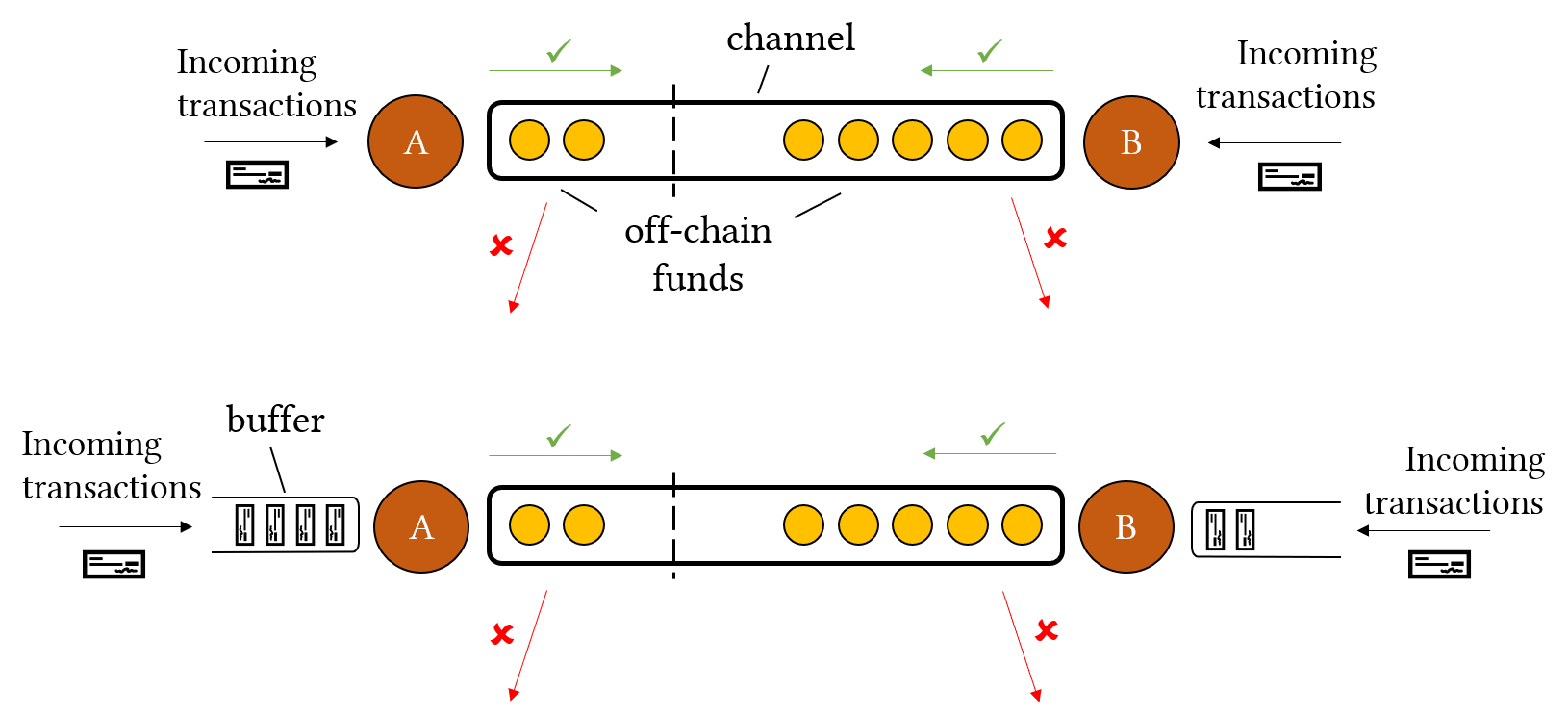}

    \caption{A payment channel without (top) and with (bottom) pending transaction buffers.}
    \label{fig:channel_drawing}
% \end{figure}
\end{wrapfigure}

Blockchain network nodes A and B can form a payment channel between themselves by signing a common commitment transaction that documents the amounts each of them commits to the channel.
For example, in the channel shown in Figure \ref{fig:channel_drawing}, node $A$'s balance in the channel is 2 coins, and node $B$'s is 5 coins.
After the initial commitment transaction is confirmed by the blockchain network, A and B can transact completely \textit{off-chain} (without broadcasting their interactions to the blockchain), by transferring the coins from one side to the other and updating their balances respectively, without the fear of losing funds thanks to a cryptographic safety mechanism.
The total funds in the channel is its capacity, which remains fixed throughout the channel's lifetime.

As nodes create multiple channels with other nodes, a network (the PCN) is formed.
In this network, if a channel does not already exist between a pair of nodes who want to transact, multihop payments are possible.
A cryptographic construct (the Hashed Time-Lock Contract -- HTLC) is again guaranteeing that the payment will either complete end-to-end, or fail for all intermediate steps.
In Figure \ref{fig:PCN_drawing} for example, node $A$ wants to pay 3 coins to node $C$, and can achieve this by paying 3 to $B$ and then $B$ paying 3 to $C$.
Another possible payment path is $A \rightarrow E \rightarrow D \rightarrow C$, which however does not have enough balance (in the $E \rightarrow D$ channel in particular) to support a payment of 3 coins.
This network creates the need for routing and scheduling of payments to achieve maximum throughput.
For more details on PCN operation, the reader is referred to \cite{Gudgeon2020a, Papadis2020}.

\paragraph{Important Metrics in PCNs}
The metrics usually used for evaluating the performance of a PCN are the payment success rate (what percentage of all transactions complete successfully), the (normalized) throughput (successful amount), and also the fees a node receives from relaying others' transactions.
A node with a lot of activity and high transacting amounts (e.g., a payment hub) might focus more on optimizing throughput, while a node transacting once in a while might care more for individual transactions succeeding.
Since fees are affine in the payment amount \cite{Papadis2020}, for a specific node to maximize the throughput of its channels is in some sense\footnote{Not strictly equivalent because of the constant term: fee = \textit{constant base fee} + \textit{proportional fee rate} $\cdot$ amount} equivalent to maximizing the fees it is earning.
Therefore, in this work we are concerned with maximizing the success rate and throughput and do not deal with fees.
Maximizing the throughput is equivalent to minimizing blockage, i.e. the amount of rejected transactions.

\begin{wrapfigure}{L}{0.3\textwidth}
% \begin{figure}[h]
    \centering
    \includegraphics[width=0.3\textwidth]{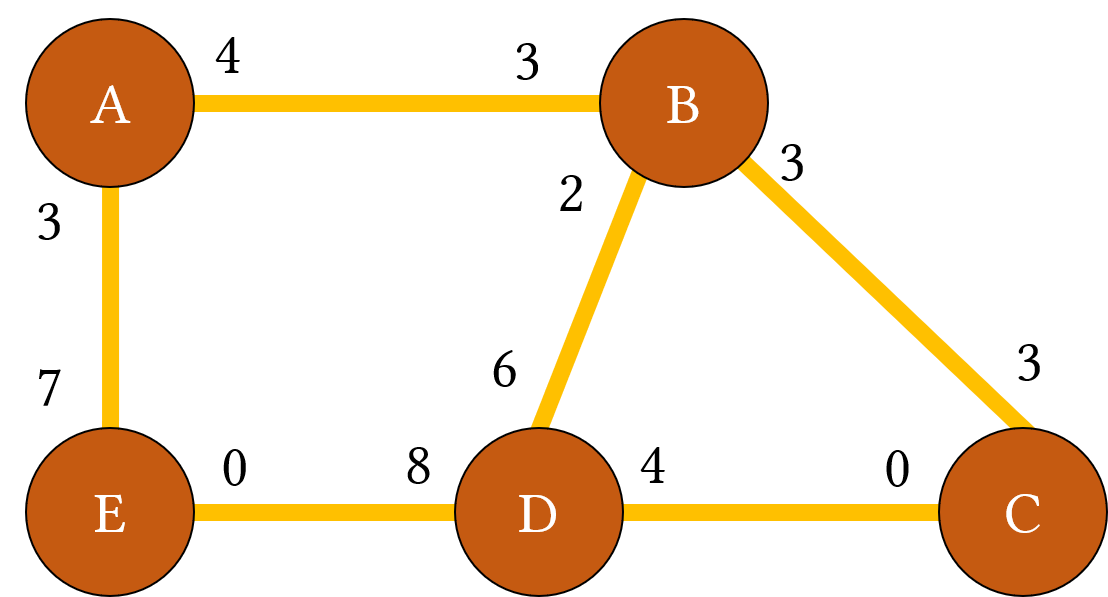}

    \caption{A payment channel network.}
    \label{fig:PCN_drawing}
% \end{figure}
\end{wrapfigure}

\paragraph{Payment scheduling policy}
The default processing mechanism in a payment channel is the following: feasible transactions are executed immediately, and infeasible transactions are rejected immediately.
In order to optimize success rates and throughput, in this work we examine whether the existence of a transaction buffer, where transactions would be pending before getting processed or rejected, would actually increase the channel performance.
We assume that the sender of every transaction (or a higher-level application which the transaction serves) specifies a deadline at most by which they are willing to wait before their transaction gets processed/rejected. 
A fine balance when choosing a deadline would be to push transactions execution to the future as much as possible in order to allow more profitable decisions within the deadline, but not too much to the extent that they would be sacrificed.
Depending on the criticality of the transaction for the application or the sender, the deadline in practice could range from a few milliseconds to a few minutes.
Note that this deadline is different than other deadlines used by the Bitcoin and Lightning protocols in time-locks (CheckLockTimeVerify -- CLTV and CheckSequenceVerify -- CSV) \cite{cltv}, as the latter are related to when certain coins can be spent by some node, while the deadline in our case refers to a Quality of Service requirement of the application.

\section{Problem formulation}
\label{sec:problem-formulation}

In this section, we introduce a stochastic model of a payment channel and define the transaction scheduling problem in a channel with buffers.

Consider an established channel between nodes $A$ and $B$ with capacity denoted by some positive natural number\footnote{All monetary quantities can be expressed in integer numbers, as in cryptocurrencies and currencies in general there exists some quantity of the smallest currency denomination, and all amounts can be expressed as multiples of this quantity. For Bitcoin, this quantity is 1 satoshi (=$10^{-8}$ bitcoins) or 1 millisatoshi.} $C$. 
Define $Q^A(t)$, $Q^B(t)$ to be the balances of nodes $A$ and $B$ in the channel at time $t$, respectively.
The capacity of a channel is constant throughout its lifetime, so obviously $Q^A(t) + Q^B(t) = C$ for all times $t \in \mathbb{R}_{+}$. We consider a continuous time model.

Transactions are characterized by their origin and destination ($A$-to-$B$ or $B$-to-$A$), their timestamp (time of arrival) $t$ and their amount $v$.
These elements are enough to describe the current channel operation in a PCN like Lightning, namely without the existence of a buffer. 
We additionally augment each transaction with a maximum buffering time, or equivalently, a deadline by which it has to be processed.
We denote the value of a transaction from $A$ to $B$ arriving at time $t_n^A$ as $v_n^A$ and its maximum buffering time as $d_n^A$ (and similarly $t_n^B, v_n^B, d_n^B$ for transactions from $B$ to $A$).
Transactions arrive at the two nodes as marked point processes: from $A$ to $B$: $\{(t_n^A, d_n^A, v_n^A)\}_{n=1}^\infty$, and from $B$ to $A$: $\{(t_n^B, d_n^B, v_n^B)\}_{n=1}^\infty$. 
Denote the deadline expiration time of the transaction as $\tau_n^A \triangleq t_n^A + d_n^A$ (similarly for B).
Denote the set of all arrival times at both nodes as $T_{\text{arrival}} = \{t_n^A\}_{n=1}^\infty \cup \{t_n^B\}_{n=1}^\infty$, and the set of all deadline expiration times as $T_{\text{expiration}} = \{\tau_n^A\}_{n=1}^\infty \cup \{\tau_n^B\}_{n=1}^\infty$.

The state of the system comprises the instantaneous balances and the contents of the buffers. 
The state at time $t$ is 
\begin{align}
\begin{split}
x(t) =~\Bigl( &Q^A(t), Q^B(t), \\
        &D_1^A(t), ..., D_{K^A(t)}^A(t), v_1^A(t), ..., v_{K^A(t)}^A(t), \\
        &D_1^B(t), ..., D_{K^B(t)}^B(t), v_1^B(t), ..., v_{K^B(t)}^B(t) \Bigr) 
\end{split}
\end{align}
where $K^A(t)$ is the number of pending transactions in node $A$'s buffer at time $t$ (similarly for $K^B(t)$),  $D_k^A(t)$ is the remaining time of transaction $k$ in node $A$'s buffer before its deadline expiration (similarly for $D_k^B(t)$), and $v_k^A(t)$ is the amount of the $k$-th transaction in node $A$'s buffer (similarly for $v_k^B(t)$).
For the channel balances, it holds that $(Q^A, Q^B) \in \{(a, b) \in [C] \times [C]: a + b = C\}$, where $[C] =  \{0,1,\dots,C\}$.
For simplicity, we assume that the pending transactions in each node's buffer are ordered in increasing remaining time order. So $D_1^A(t) \leq D_2^A(t) \leq ... \leq D_{K^A(t)}^A(t)$, and similarly for $B$.

A new arriving transaction causes a transition to a state that includes the new transaction in the buffer of the node it originated from. 
The evolution of the system is controlled, with the controller deciding whether and when to serve each transaction.
At time $t$, the set of possible actions at state $x(t)$ is a function of the state and is denoted by $U(x(t))$.
Specifically, a control policy at any time $t$ might choose to process (execute) some transactions and drop some others.
When a transaction is processed or dropped, it is removed from the buffer where it was stored.
Additionally, upon processing a transaction the following balance updates occur:
\begin{align*}
    (Q^A, Q^B) &\rightarrow (Q^A - v, Q^B + v), &   \text{if the processed transaction is from A to B and of amount $v$} \\
    (Q^A, Q^B) &\rightarrow (Q^A + v, Q^B - v), &   \text{if the processed transaction is from B to A and of amount $v$}
\end{align*}

At time $t$, the allowable actions $u(t)$ are subsets of the set $U'(t) = \{ (node, k, action): node \in \{A, B\}, 1 \leq k \leq K^{node}(t), action \in \{EX,DR\} \}$ that contain transactions in a specific order such that executing and dropping them in that order is possible given the channel state at time $t$.
Action $EX$ means ``execute the transaction,'' while action $DR$ means ``drop the transaction.''
Formally, 
\begin{align}
\begin{split} 
u(t) &\in U(x(t)) = \\
    \bigl\{ &u = \{(node_i, k_i, action_i)\}_{i=1}^l \in \mathcal{P}(U'(t)): \\
    &\forall i = 1, \dots, l, action_i \text{ on the $k_i$-th transaction of $node_i$ is feasible after applying} \\
    &\text{the first $i-1$ actions on the respective transactions} \bigr\}
\end{split}
\end{align}
where $\mathcal{P}$ denotes the powerset of a set.
Note that the empty set is also an allowable action and means that at that time the control policy idles (i.e. neither processes nor drops any transaction).
An expiring transaction that is not processed at the time of its expiration is automatically included in the dropped transactions at that time instant.

Having defined all the possible actions, we should note the following: in the presence of a buffer, more than one transaction might be executed at the same time instant, either because two or more transactions expire at that time, or because the policy decides to process two or more. 
The total amount processed (if $action = EX$) or dropped (if $action = DR$) by the channel at time $t$ is:
% \begin{equation}
% \Tilde{v}_{EX}^{u(t)}(t) = \sum_{(k,node,EX) \in u(t)} v_k^{node} (t)
% \end{equation}
% and the total amount rejected by the channel at time $t$ is:
% \begin{equation}
% \Tilde{v}_{DR}^{u(t)}(t) = \sum_{(k,node,DR) \in u(t)} v_k^{node} (t)
% \end{equation}
\begin{equation}
\Tilde{v}_{action}^{u(t)}(t) = \sum_{(k,node,action) \in u(t)} v_k^{node} (t)
\end{equation}

For example, if $u(t) = \{ (A, 2, EX), (B, 3, DR), (B, 1, EX) \}$ (meaning that at time $t$ the chosen action is to execute the second transaction from the buffer of node $A$, drop the third transaction from the buffer of node $B$, and execute the first transaction from the buffer of node $B$), then $\Tilde{v}_{EX}^{u(t)}(t) = v_2^A + v_1^B$ and $\Tilde{v}_{DR}^{u(t)}(t) = v_3^B$.

A control policy $\pi = \{(t, u(t))\}_{t \in \mathbb{R}_{+}}$ consists of the times $t$ and the corresponding actions $u(t)$, and belongs to the set of admissible policies
\begin{equation}
\Pi = \bigl\{ \{(t, u(t))\}_{t \in \mathbb{R}_{+}} \text{ such that } u(t) \in U(x(t)) \text{ for all } t \in \mathbb{R}_{+} \bigr\}
\end{equation}
% where $U(x(t)) = \mathcal{P}(U^A(t) \cup U^B(t) \cup \{N\})$ and $\mathcal{P}$ denotes the powerset of a set.

The total amount of transactions that have arrived until time $t$ is
\begin{align}
V_{\text{total}}(t) = \sum_{\substack{n \in \mathbb{N}:~t_n \leq t}} v_{n}
\end{align}

The total throughput (i.e. volume of successful transactions) up to time $t$ under policy $\pi$ is:
\begin{align}
S^{\pi}(t) = \int_{\tau = 0}^t \Tilde{v}_{EX}^{u(\tau)}(\tau) d\tau
\end{align}

The total blockage (i.e. volume of rejected transactions) up to time $t$ under policy $\pi$ is:
\begin{align}
R^{\pi}(t) = \int_{\tau = 0}^t \Tilde{v}_{DR}^{u(\tau)}(\tau) d\tau
\end{align}

The amount of pending transactions under policy $\pi$ is then the difference between the total amount and the sum of the successful and rejected amounts:
\begin{align}
\label{eqn:equivalance-of-S-and-R-1}
P^{\pi}(t) = V_{\text{total}}(t) - S^{\pi}(t) - R^{\pi}(t)
\end{align}

% The expected long-term average throughput under policy $\pi$ is defined as
% \begin{equation}
% J^{\pi}(x(0)) \triangleq \lim_{T \rightarrow \infty} \frac{1}{T} \expect{S^{\pi}(T)}
% \end{equation}

% An optimal policy $\pi^*$ is one that satisfies $J^{\pi^*}(x(0)) \geq J^{\pi}(x(0))$ for all $\pi \in \Pi$.

The objective is to maximize the total channel throughput (or minimize the total channel blockage) over all admissible dynamic policies.

A few final notes on the assumptions: We assume that both nodes have access to the entire system state, namely to the buffer contents not only of themselves, but also of the other node in the channel. 
% This is reasonable to assume, as if the proposed modification of buffer addition with shared knowledge indeed improves the channel throughput (as we claim and show), then nodes will have the incentive to switch to such a scheme and reveal their buffer contents to their counterparty (which for the executed transactions would anyway be revealed at the time of their execution). 
Therefore, in our model, referring to one buffer per node or to a single shared buffer between the nodes is equivalent.
Moreover, our implicit assumption throughout the paper is that the buffer sizes are not constrained.
This implies that allowing or disallowing a ``Drop'' action does not make a difference in terms of the optimality a policy can achieve.
To see this, suppose that node $A$ wants to drop a transaction at some time before its expiration deadline, including its arrival time.
What $A$ can do is wait until the transaction's expiration without processing it, and then it will automatically expire and get dropped.
This has the same effect as dropping the transaction earlier.
% Therefore, a ``Drop'' action does not give add or remove any flexibility from an optimal policy.
Although a ``Drop'' action does not give add or remove any flexibility from an optimal policy, it is helpful for simplifying the proof of Lemma \ref{lemma:channel-with-buffers}, and so we adopt it.
If, however, the buffer sizes are limited, then the need for nodes to select which transactions to keep pending in their buffers arises, and dropping a transaction as soon as it arrives or at some point before its expiration deadline might actually lead to a better achieved throughput.
As this case likely makes the problem combinatorially difficult, we do not consider it in the present work.

The notation defined so far is summarized in Table \ref{table:notation} in Appendix \ref{app:summary-of-notation}.

\section{Throughput-optimal scheduling in a payment channel}
\label{sec:policy-section}

In this section, we determine a scheduling policy for the channel and prove its optimality. % with respect to the total blockage and expected long-term average throughput.
The policy takes advantage of the buffer contents to avoid dropping infeasible transactions by compensating for them utilizing transactions from the opposite side's buffer.

We first note that buffering does not only apply to transactions that are infeasible on arrival, as in done for example in \cite{Sivaraman2020}.
An example where buffering even transactions that are feasible at their time of arrival and not processing them right away can actually improve the success rate and the throughput is shown in Figure \ref{fig:example}.
At $t = 0$, node $A$ has a balance of $Q^A(0) = 7$, and two transactions from A to B in its buffer, with remaining times and values as follows:
$(D_1^A(0), v_1^A) = (3,9), (D_2^A(0), v_2^A) = (5,2)$.
At $t = 1$, a transaction of amount 2 from B to A arrives and is processed immediately.
At $t = 4$, another transaction of amount 2 from B to A arrives and is processed immediately.
Now consider the two cases: 
\begin{itemize}
    \item If the transaction (5,2) is executed at $t = 0$, then the transaction (3,9) will be rejected. 
    In this case, at $t = 5$ the number of successful transactions is 3 out of 4, and the throughput is 6.

    \item If the transaction (5,2) waits until its deadline (which expires at $t = 5$), then both (5,2) and (3,9) will go through. 
    In this case, at $t = 5$ the number of successful transactions is 4 out of 4, and the throughput is 15.
\end{itemize}
Therefore, although (5,2) is feasible at the time of its arrival, not processing it directly and placing it into the buffer for subsequent processing (as done in the second case) leads to more transactions being executed and higher throughput eventually.

\begin{figure}[h]
    \centering
    \includegraphics[scale=0.4]{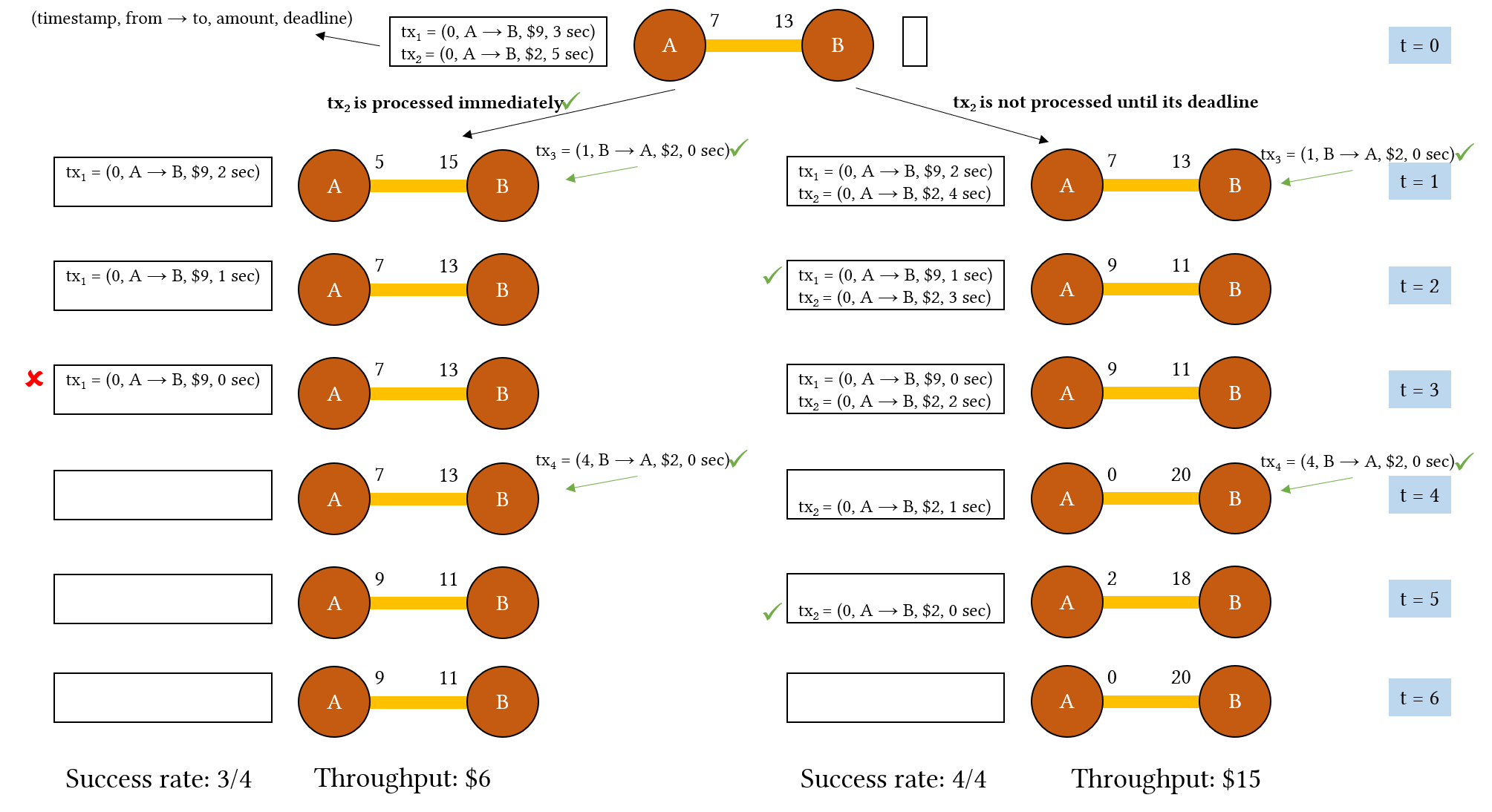}
    \caption{An example demonstrating that buffering even transactions feasible at the time of their arrival can increase the success rate and the throughput.}
    \Description{}
    \label{fig:example}
\end{figure}

Although the benefit from buffering transactions is intuitive, in the general case where arriving transaction amounts are allowed to vary, finding an optimal policy is intractable.
Specifically, for a single channel without buffers and for transactions of varying amounts, finding an optimal policy that maximizes the number of transactions executed (equivalently, the success rate) is NP-hard.
An offline version of this problem with a finite input is defined in \cite{ACD}: $N$ transactions $\{(t_n^{A/B}, v_n^{A/B})\}_{n=1}^N$ of monetary value $v_n$ arrive at times $t_n$ from either side, and the goal is to find a subset of the arriving transactions to be executed in the order of arrival that maximizes the number of successful executions.
(To see how this problem fits in our case, consider our more general model of section \ref{sec:problem-formulation} with all buffering times equal to zero).
The decision version of the problem is proven (as Problem 2 with proof in section 3.2 of \cite{ACD}) to be NP-complete.
% and since an optimal policy cannot be found, the authors provide a polynomial approximation algorithm. 
Therefore, finding an optimal policy in the general online setting of a single channel with possibly infinite input of transactions is intractable.
We expect that the same is true when the objective is to maximize the total throughput.
For this reason, in the theoretical part of the paper we focus our attention on the online case of a single channel with equal amounts for all arriving transactions, for which an optimal policy can be analytically found.

\subsection{General case: channel with buffers}
\label{sec:general-case-channel-with-buffers}

We define policy PMDE (Process or Match on Deadline Expiration) for scheduling transactions in the fixed amounts case.
The optimality of PMDE will be shown in the sequel and is the main result of this paper.

% To define PMDE formally, consider the times of deadline expirations $T_{\text{expiration}} = \{\tau_n^A\}_{n=1}^\infty \cup \{\tau_n^B\}_{n=1}^\infty$. If at time $\tau_n^{A-}$ the $n$-th transaction from $A$ to $B$ is still in the system, then if $Q^A(\tau_n^{A-}) \geq v_n^A$, then execute the transaction. 
% If $Q^A(\tau_n^{A-}) < v_n^A$ and $Q^B(\tau_n^{A-}) > v_n^A$ and $K^B(\tau_n^{A-}) \geq 1$, then execute the transaction with remaining time $D_1^B(\tau_n^{A-})$ from $B$ to $A$, and right after execute the expiring transaction from $A$ to $B$.

\begin{algorithm}[ht]
% \SetAlgoLined
\LinesNumbered
\SetKwFor{On}{on}{do}{}
% \KwOut{???}
\KwIn{channel state (balances and buffer contents)}

\On{arrival of transaction $p_n^A$ at time $t_n^A$} {
    add $p_n^A$ to A's buffer
}

\On{deadline expiration of transaction $p_n^A$ at time $\tau_n^A$} {
    \uIf{$p_n^A$ is in A's buffer at time $\tau_n^{A-}$}{
        \If{$Q^A(\tau_n^{A-}) \geq v_n^A$}{
            execute $p_n^A$\;
        }
        \uElseIf{$Q^A(\tau_n^{A-}) < v_n^A$ and $Q^B(\tau_n^{A-}) \geq v_n^A$ and $K^B(\tau_n^{A-}) \geq 1$}{
            execute the transaction with remaining time $D_1^B(\tau_n^{A-})$ from $B$ to $A$\;
            execute $p_n^A$\;
        }
        \Else{
            drop $p_n^A$\;
        }
    }
    \Else{
        idle
    }
}
\caption{PMDE scheduling policy (Process or Match on Deadline Expiration)}
\label{alg:PMDE-policy}
\end{algorithm}
The policy is symmetric with respect to nodes A and B.

In words, PMDE operates as follows:
Arriving transactions are buffered until their deadline expires. 
On deadline expiration (actually just before, at time $\tau_n^{A-}$), if the expiring transaction is feasible, it is executed. 
If it is not feasible and there are pending transactions in the opposite direction, then the transaction with the shortest deadline from the opposite direction is executed, followed immediately by the execution of the expiring transaction. 
Otherwise, the expiring transaction is dropped.

% This is reasonable to assume, as if the proposed modification of buffer addition with shared knowledge indeed improves the channel throughput (as we claim and show), then nodes will have the incentive to switch to such a scheme and reveal their buffer contents to their counterparty (which for the executed transactions would anyway be revealed at the time of their execution). 
Note that the only information sharing between the two nodes PMDE requires is the expiring transaction(s) at the time of expiration, information which would be revealed anyway at that time.
So PMDE is applicable also for nodes not willing to share their buffer's contents.

In the general case of non-fixed transaction amounts, the greedy policy PMDE is not optimal for either objective.
This is shown in the following counterexample.
Consider a channel with balance $10$ at node $A$ and one big transaction of amount $9$ and $5$ small transactions of amounts $2$ arriving in this order from node $A$ to node $B$.
If the big one, which is feasible, is processed greedily immediately, then the small ones become infeasible. The total success rate in this case is $1/6$ and the total throughput is $9$.
While if the big one is rejected, then all the small ones are feasible. The total success rate in this case is $5/6$ and the total throughput is $10$. 
So PMDE is not optimal when transaction amounts are unequal, neither with respect to the success rate, nor with respect to the throughput.

We now proceed to show PMDE's optimality in the equal transaction amount case.
Note that in this case, the objectives of maximizing the success rate and maximizing the throughput are equivalent, as they differ only by a scaling factor (the transaction value divided by the total number of transactions), and have the same maximizing policy.
Note also that combining transactions from the two sides as PMDE does requires that at least one of the transactions is individually feasible.
This will always happen as long as $Q^A(0) \geq v$ or $Q^B(0) \geq v$ in the fixed amounts case\footnote{Even in the general non-fixed amounts case though, the chance of two transactions individually infeasible, that is with amounts larger than the respective balances, occurring in both sides of the channel simultaneously is very small: usually, the transaction infeasibility issue is faced at one side of the channel because the side is depleted and funds have accumulated on the other side.}.

This optimality of PMDE with respect to blockage is stated in Theorem \ref{thm:rejection-optimality}, the main theorem of this paper.
% In addition, Corollary \ref{cor:long-term-throughput-optimality} also tells us that 
This blockage optimality of PMDE also implies its expected long-term average throughput optimality.

\vspace{5pt}

\begin{thmboxed}
\label{thm:rejection-optimality}
For a payment channel with buffers under the assumption of fixed transaction amounts, let $R$ be the total rejected amount when the initial state is $x(0)$ and transaction are admitted according to a policy $\pi \in \Pi$, and $R^{PMDE}$ the corresponding process when PMDE is applied instead. 
Then, for any sample path of the arrival process, it holds
\begin{equation}
    R^{PMDE}(t) \leq R^{\pi}(t) \text{ a.s. for all } t \in \mathbb{R}_{+}
\end{equation}
\end{thmboxed}

\vspace{5pt}

We would like PMDE to be maximizing the channel throughput among all dynamic policies.
However, this is not true at every time instant.
To see this, consider another policy that ignores the existence of the buffer and processes transactions immediately as soon as they arrive if they are feasible and drops no transactions, and assume the channel balances are big enough that for some time no transaction is infeasible. 
Then this policy achieves higher throughput in the short term compared to PMDE, as PMDE waits until the deadline expiration to execute a feasible transaction, while the other policy executes it right away.
For example, up to the first deadline expiration, assuming at least one transaction up to then is feasible, the other policy achieves nonzero throughput while PMDE achieves zero throughput.

Therefore, the optimality of PMDE does not hold for the throughput at every time instant. 
It holds for another quantity though: the total blockage (and because of \eqref{eqn:equivalance-of-S-and-R-1}, it also holds for the sum of the successfully processed amounts plus the pending ones).

Let $\Pi^{DE}$ be the class of dynamic policies that take actions only at the times of Deadline Expirations.
We will first prove that to minimize blockage it suffices to restrict our attention to policies in $\Pi^{DE}$.
This is shown in the following lemma.

\begin{lemma}
\label{lemma:DE-policies-suffice}
For every policy $\pi \in \Pi$, there exists another policy $\Tilde{\pi} \in \Pi^{DE}$ that take actions only at the times of deadline expirations, and the states and blockage at the times of deadline expirations under $\Tilde{\pi}$ are the same as under $\pi$:
\begin{equation}
    \Tilde{x}(\tau) = x(\tau)~\text{ and }~\Tilde{R}(\tau) = R(\tau)
\end{equation}
% and
% \begin{equation}
%   \Tilde{R}(\tau) = R(\tau)
% \end{equation}
for all $\tau \in T_\text{expiration}$, and for any sample path of the arrival process.
\end{lemma}

\begin{proof}
% Let $\pi \in \Pi$ be an arbitrary policy that during the interval $[0,\tau_1]$ drops all the transactions in the set $P_1^d$ and processes all the $L$ transactions in the set $P_1^p = \{p_1, \dots, p_L\}$ in the order $p_1, \dots, p_L$ at times $t_{p_1}, \dots, t_{p_L}$.
% Denote the state update function by $f:X \times U \rightarrow X$, as it was described in detail in section \ref{sec:problem-formulation}.

Let $\pi \in \Pi$ be an arbitrary policy that during the interval $[0,\tau_1]$ drops certain transactions and processes certain other transactions in some specific order.
We define another policy that takes no action during $[0,\tau_1)$ and at $\tau_1$ processes and drops the same transactions that $\pi$ has processed and dropped respectively during $[0,\tau_1]$, in exactly the same order.
This is possible, since $\tau_1$ is the first expiration time.
Thus, at $\tau_1$ we have that the states (balances and buffer contents) and blockages under $\pi$ and $\Tilde{\pi}$ are exactly the same.
Now, defining $\Tilde{\pi}$ analogously and applying the same argument on the intervals $(\tau_1, \tau_2], (\tau_2, \tau_3], \dots$ inductively proves the lemma.
\end{proof}

To prove Theorem \ref{thm:rejection-optimality}, we will also need the following lemma.

\begin{lemma}
\label{lemma:channel-with-buffers}
For every policy $\pi \in \Pi^{DE}$, there exists a policy $\Tilde{\pi} \in \Pi^{DE}$ that acts similarly to PMDE at $t = \tau_1$ and is such that when the system is in state $x(0)$ at $t=0$ and policies $\pi$ and $\Tilde{\pi}$ act on it, the corresponding total rejected amount processes $R$ and $\Tilde{R}$ can be constructed via an appropriate coupling of the arrival processes so that
\begin{equation}
\label{eqn:lemma-with-buffers-main-blockage}
\Tilde{R}(t) \leq R(t), t \in \tau_1, \tau_2, \dots
\end{equation}
\end{lemma}

The proof idea is the following: 
We construct $\Tilde{\pi}$ and couple the blockage processes under $\pi$ and $\Tilde{\pi}$ and identical transaction arrival processes so that \eqref{eqn:lemma-with-buffers-main-blockage} holds.
First, we consider what policies $\pi$ and $\Tilde{\pi}$ might do at time $\tau_1$ of the first deadline expiration.
Then, for each possible combination, we couple $\Tilde{\pi}$ with $\pi$ in subsequent times so that at some point the states (balances and buffer contents) and the total blockages under $\pi$ and $\Tilde{\pi}$ coincide, and so that \eqref{eqn:lemma-with-buffers-main-blockage} is being satisfied at all these times.
From then on, we let the two policies move together.
The full proof is given in Appendix \ref{app:proof-of-main-lemma}.

Next, we present the proof of Theorem \ref{thm:rejection-optimality}.

\begin{proof}

The proof proceeds as follows: 
We first use Lemma \ref{lemma:DE-policies-suffice} to say that a blockage-minimizing policy among all policies in $\Pi$ exists in the class $\Pi^{DE}$.
We then repeatedly use Lemma \ref{lemma:channel-with-buffers} to construct a sequence of policies converging to the optimal policy.
Each element of the sequence matches the optimal policy at one more step each time, and is at least as good as any other policy until that point in time. 
Having acquired this sequence of policies that gradually tends to the proposed optimal policy, we can inductively show that the proposed optimal policy PMDE achieves higher throughput than any other policy.
A similar technique is used in sections IV and V of \cite{tinfo}.

From Lemma \ref{lemma:DE-policies-suffice}, we have that for any policy $\pi \in \Pi$, we can construct another policy $\pi' \in \Pi^{DE}$ such that for the corresponding total blockage processes $R^{\pi}$ and $R^{\pi'}$ we have $R^{\pi'}(t) \leq R^{\pi}(t)$, $t \in \mathbb{R}_{+}$.

From Lemma \ref{lemma:channel-with-buffers}, we have that given policy $\pi' \in \Pi^{DE}$, we can construct a policy $\pi_1 \in \Pi^{DE}$ that is similar to PMDE at $t=\tau_1$ and is such that for the corresponding total blockage processes $R^{\pi'}$ and $R^{\pi_1}$ we have $R^{\pi_1}(t) \leq R^{\pi'}(t)$, $t \in T_{\text{expiration}}$.

By repeating the construction, we can show that there exists a policy $\pi_2$ that agrees with $\pi_1$ at $t=\tau_1$, agrees with PMDE at $t=\tau_2$, and is such that for the corresponding total blockage processes we have $R^{\pi_2}(t) \leq R^{\pi_1}(t)$, $t \in T_{\text{expiration}}$.

If we repeat the argument $k$ times, we obtain policies $\pi_i$, $i=1,\dots,k$, such that policy $\pi_i$ agrees with PMDE up to and including $\tau_i$, and for the for the corresponding total blockage processes we have $R^{\pi_k}(t) \leq R^{\pi_{k-1}}(t) \leq \dots \leq R^{\pi_1}(t) \leq R^{\pi'}(t) \leq R^{\pi}(t)$, $t \in T_{\text{expiration}}$.

Taking the limit as $k \rightarrow \infty$:
\begin{align}
    \lim_{k \rightarrow \infty} R^{\pi_k}(t) &= R^{PMDE}(t) %\\
    % \pi_k &\rightarrow PMDE \text{ a.s. as } k \rightarrow \infty
\end{align}
Therefore, $R^{PMDE}(t) \leq R^{\pi}(t)$, $t \in T_{\text{expiration}}$.
\end{proof}

Note that the proven optimality results hold independently of the capacity and initial balances.

\subsection{Special case: channel without buffers}
\label{sec:special-case-channel-without-buffers}

\subsubsection{Optimal policy for the channel without buffers}

The results of section \ref{sec:general-case-channel-with-buffers} apply also in the special case where either buffers are nonexistent (and therefore all transactions have to processed or dropped as soon as they arrive), or when all buffering times of arriving transactions are zero.
In this case, deadline expiration times are the same as arrival times, and policy PMDE becomes the following policy PFI (= Process Feasible Immediately):

\begin{algorithm}[ht]
% \SetAlgoLined
\LinesNumbered
\SetKwFor{On}{on}{do}{}
% \KwOut{???}
\KwIn{channel state (balances only)}

\On{arrival of transaction $p_n^A$ at time $t_n^A$} {
    \If{$Q^A(t_n^{A-}) \geq v_n^A$}{
    execute $p_n^A$\;
    }
    \Else{
    drop $p_n^A$\;
    }
}
\caption{PFI scheduling policy (Process Feasible Immediately)}
\label{alg:PFI-policy}
\end{algorithm}

In words, upon transaction arrival, PFI executes the transaction immediately if it is feasible, and drops the transaction immediately if it is not feasible.
Formally, PFI takes action $(A,1,EX)$ at all times $t_n^A$, $n \in \mathbb{N}$, if $Q^A(t_n) \geq v_n^A$ and action $(A,1,DR)$ otherwise, action $(B,1,EX)$ at all times $t_n^B$, $n \in \mathbb{N}$, if $Q^B(t_n) \geq v_{t_n}^B$ and action $(B,1,DR)$ otherwise.

The following corollary states the analog of Theorem \ref{thm:rejection-optimality} and for the case of the channel without buffers.

\begin{corollary}
% \label{thm:optimal-policy-without-buffer-long-term}
For a single channel without buffers under the assumption of fixed transaction amounts, policy PFI is optimal with respect to the total blockage: %and the expected long-term average throughput:
Let $R$ be the total rejected amount when the initial state is $x(0)$ and transaction are admitted according to a policy $\pi \in \Pi$, and $R^{PMDE}$ the corresponding process when PMDE is applied instead. Then, for any sample path of the arrival process, it holds
\begin{equation}
\label{eqn:sample-path-blockage-optimality-no-buffers}
    R^{PMDE}(t) \leq R^{\pi}(t) \text{ a.s. for all } t \in \mathbb{R}_{+}
\end{equation}
% and
% \begin{equation}
%   J^{PFI} \geq J^{\pi}
% \end{equation}

In addition, in this case the following also holds for any sample path of the arrival process:
\begin{equation}
\label{eqn:sample-path-throughput-optimality-no-buffers}
S^{PFI}(t) \geq S^{\pi}(t) \text{ a.s. for all } t \in \mathbb{R}_{+}
\end{equation}
\end{corollary}

Equation \eqref{eqn:sample-path-throughput-optimality-no-buffers} is a direct consequence of \eqref{eqn:sample-path-blockage-optimality-no-buffers} and \eqref{eqn:equivalance-of-S-and-R-1}, as in this case the pending transaction amount is always zero.

\subsubsection{Analytical calculation of optimal success rate and throughput for the channel without buffers}

For a channel without buffers, if the arrivals follow a Poisson process, we can calculate the optimal success rate and throughput as the ones we get by applying the optimal policy PFI.

\begin{theorem}
\label{thm:analytical-calculation-of-success-rate}
For a single channel between nodes $A$ and $B$ with capacity $C$, and Poisson transaction arrivals with rates $\lambda_A \neq \lambda_B$ and fixed amounts equal to $v$, the maximum possible success rate of the channel is
\begin{equation}
SR_{\text{opt}} = \lambda_A \left( 1 - \frac{\lambda_B/\lambda_A - 1}{(\lambda_B/\lambda_A)^{\Tilde{C}+1} - 1} \right) 
+ \lambda_B \left( 1 - \left(\frac{\lambda_B}{\lambda_A}\right)^{\Tilde{C}} \frac{\lambda_B/\lambda_A - 1}{(\lambda_B/\lambda_A)^{\Tilde{C}+1} - 1} \right)
\end{equation}
where $\Tilde{C} = \lfloor \frac{C}{v} \rfloor$.\\

When $\lambda_A = \lambda_B = \lambda$, the maximum possible success rate is
\begin{equation}
SR_{\text{opt}} = \frac{2\lambda \Tilde{C}}{\Tilde{C}+1}
\end{equation}
\end{theorem}

\noindent
A proof of this result is given in Appendix \ref{app:proof-of-analytical-calculation}.
The maximum possible normalized throughput is $S \cdot v$.

\section{Heuristic policies for general amount distributions}
\label{sec:heuristic-policies}

So far, we have described our PMDE policy and proved its optimality for a channel with or without buffering capabilities in the case of fixed arriving transaction amounts.
PMDE could also serve though the more general case of arbitrary amounts if payment splitting is used.
Indeed, there have been proposals (e.g., \cite{Sivaraman2020}) that split payments into small chunks-packets and route or schedule them separately, possibly along different paths and at different times, utilizing Atomic Multipath Payments (AMP) \cite{AMP}.
Recall that it is guaranteed by the PCN's cryptographic functionality (the HTLC chaining) that a multihop payment will either complete or fail along all intermediate steps.
The additional constraint if AMP is employed is that some check should be performed to ensure that all chunks of a particular transaction are processed until their destination, or all chunks are dropped.
This could for example be checked when the transaction deadline expires, at which moment every node would cancel all transactions for which it has not received all chunks.
Therefore, this is one way to be able to apply PMDE to an arbitrary transaction amounts setting.

We also present a modified version of PMDE that does not require payment splitting and AMP, and is a heuristic extension of the policy that was proved optimal for fixed transaction amounts.
Since now a transaction of exactly the same amount as the expiring one is unlikely to exist and be the first in the opposite node's buffer, the idea is to modify the matching step of PMDE so that the entire buffer of the opposite side is scanned until enough opposite transactions are found so as to cover the deficit.
The buffer contents are sorted according to the criterion \textit{bufferDiscipline}, possible values for which are: oldest-transaction-first, youngest-transaction-first, closest-deadline-first, largest-amount-first, smallest-amount-first.
The modified policy PMDE is shown in Algorithm \ref{alg:PMDE-augmented-policy}, and is symmetric with respect to nodes A and B.

\begin{algorithm}[ht]
% \SetAlgoLined
\LinesNumbered
\SetKwInOut{Parameter}{Parameters}
\SetKwFor{On}{on}{do}{}
% \KwOut{???}
\KwIn{channel state (balances and buffer contents)}
\Parameter{bufferDiscipline}

\On{arrival of transaction $p_n^A$ at time $t_n^A$} {
    add $p_n^A$ to A's buffer
}

\On{deadline expiration of transaction $p_n^A$ at time $\tau_n^A$} {
    \uIf{$p_n^A$ is in A's buffer at time $\tau_n^{A-}$}{
    \If{$Q^A(\tau_n^{A-}) \geq v_n^A$}{
            execute $p_n^A$\;
        }
        \uElseIf{$Q^A(\tau_n^{A-}) < v_n^A$ and $Q^B(\tau_n^{A-}) \geq v_n^A$ and $K^B(\tau_n^{A-}) \geq 1$}{
            deficit $\leftarrow Q^A(\tau_n^{A-}) - v_n^A$\;
            % totalOppositeAmount $\leftarrow 0$\;
            scan transactions in B's  buffer in order \textit{bufferDiscipline} and find the first set with total amount > deficit\;
            % \While{totalOppositeAmount < deficit}{
            % }
            \If{such a set exists}{
            execute these transactions from $B$ to $A$\;
            execute $p_n^A$\;
            }
            \Else{
            drop $p_n^A$\;
            }
        }
        \Else{
            drop $p_n^A$\;
        }
    }
    \Else{
        idle
    }
}
\caption{Generalized PMDE scheduling policy}
\label{alg:PMDE-augmented-policy}
\end{algorithm}

We also evaluate another family of heuristic policies that sort the transactions in the buffer according to some criterion and process them in order, and which is shown in Alg. \ref{alg:heuristic-policies}.
The buffer might be shared between the two nodes (thus containing transactions in both directions), or separate.
At regular intervals (every \textit{checkInterval} seconds -- a design parameter), expired transactions are removed from the buffer. 
Then, the buffer contents are sorted according to the criterion \textit{bufferDiscipline}, and as many of them as possible are processed by performing a single linear scan of the sorted buffer.
The policies of this family are also parameterized by \textit{immediateProcessing}. 
If \textit{immediateProcessing} is true, when a new transaction arrives at a node, if it is feasible, it is processed immediately, and only otherwise added to the buffer, while if \textit{immediateProcessing} is false, all arriving transactions are added to the buffer regardless of feasibility.
The rationale behind non-immediate processing is that delaying processing might facilitate the execution of other transactions that otherwise would not be possible to process.

\begin{algorithm}[ht]
% \SetAlgoLined
\LinesNumbered
\SetKwInOut{Parameter}{Parameters}
\SetKwFor{On}{on}{do}{}
\SetKwFor{Every}{every}{do}{}
% \KwOut{???}
\KwIn{channel state (balances and buffer contents)}
\Parameter{bufferDiscipline, immediateProcessing, checkInterval}

\On{arrival of transaction $p_n^A$ at time $t_n^A$} {
    \If{immediateProcessing = True and $Q^A(t_n^{A-}) \geq v_n^A$}{
        execute $p_n^A$
    }
    \Else{
        add $p_n^A$ to A's buffer
    }
}

\Every{checkInterval} {
    remove expired transactions from buffer\;
    sortedBuffer $\leftarrow$ sort(buffer, bufferDiscipline)\;
    \For{transaction $p \in \text{sortedBuffer}$}{
        \If{$p$ is feasible}{
            execute $p$\;
        }
    }
}
\caption{PRI scheduling policy (Process at Regular Intervals)}
\label{alg:heuristic-policies}
\end{algorithm}

An underlying assumption applying to all policies is that the time required for buffer processing is negligible compared to transaction interarrival times.
Thus, buffer processing is assumed effectively instantaneous.

\section{Evaluation}
\label{sec:evaluation-results}

\subsection{Simulator}
\label{sec:simulator}

In order to evaluate the performance of different scheduling policies, especially in the analytically intractable case of arbitrary transaction amounts, we built a discrete event simulator of a single payment channel with buffer support using Python SimPy \cite{simpy}. 
Discrete Event Simulation, as opposed to manual manipulation of time, has the advantages that transaction arrival and processing occur as events according to user-defined distributions, and the channel is a shared resource that only one transaction can access at a time.
Therefore, such a discrete event simulator captures a real system's concurrency and randomness more realistically.
The simulator allows for parameterization with respect to the initial channel balances, the transaction generation distributions (frequency, amount, and maximum buffering time) for both sides of the channel, and the total transactions to be simulated.
The channel has two buffers attached to it that operate according to the scheduling policy being evaluated.
The code of our simulator 
% is available in \cite{single-payment-channel-simulator-code} 
will be open-sourced.

\subsection{Experimental setup}
\label{sec:experimental-setup}

\begin{wrapfigure}{r}{0.3\textwidth}
    \centering
    \includegraphics[width=0.3\textwidth]{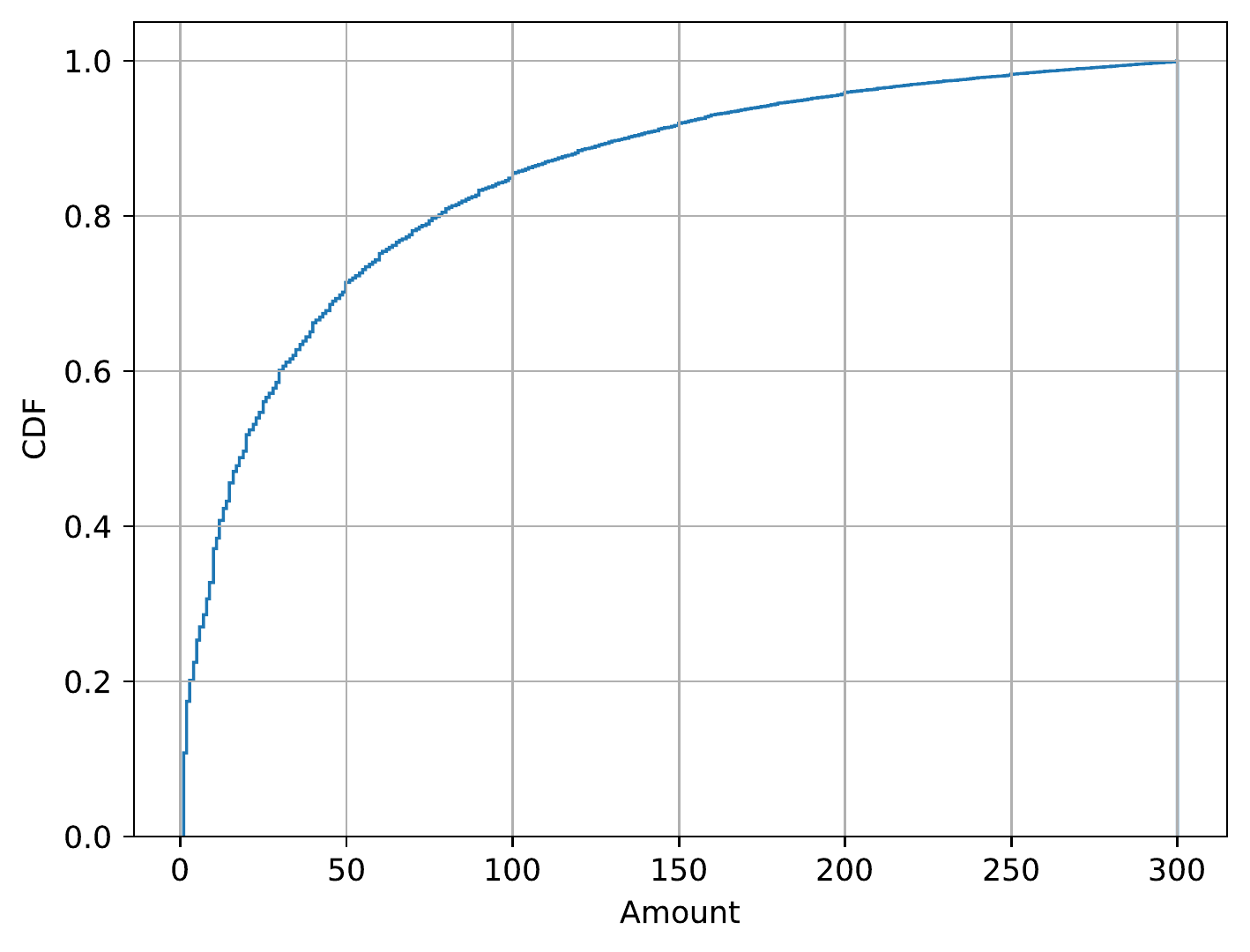}

    \caption{Cumulative Distribution Function of the amounts used from the credit card transaction dataset \cite{kaggle-dataset}.}
    \label{fig:dataset_cdf}
\end{wrapfigure}

\subsubsection{Optimal policy for arbitrary amounts}

We simulate a payment channel between nodes 0 and 1 with a capacity of 300 with initial balances of 0 and 300 respectively\footnote{In practice, Lightning channels are usually single-funded initially \cite{Pickhardt2019}.}.
Transactions are arriving from both sides according to Poisson distributions.
We evaluate policies PMDE and PRI defined in section \ref{sec:heuristic-policies}, with (PRI-IP) or without immediate processing (PRI-NIP), for all 5 buffer disciplines (so 15 policies in total), and when both nodes have buffering capabilities (with and without shared knowledge of the contents), or only one node, or none.
Each experiment is run for around 1500 seconds, and in the results only the transactions that arrived in the middle 80\% of the total simulation time are accounted for, so that the steady-state behavior of the channel is captured.
Unless otherwise stated, we present our results when using the oldest-transaction-first buffer discipline, and the \textit{checkInterval} parameter used by the PRI policies is set to 3 seconds.
For studying the general amounts case, we use synthetic data from Gaussian and uniform distributions, as well as an empirical distribution drawn from credit card transaction data.
The dataset we used is \cite{kaggle-dataset} (also used in \cite{Sivaraman2020}) and contains transactions labeled as fraudulent and non-fraudulent.
We keep only the latter, and from those we sample uniformly at random among the ones that are of size less than the capacity.
The final distribution we draw from is shown in Figure \ref{fig:dataset_cdf}.
% , and the raw data can be found in \cite{single-payment-channel-simulator-code}.
Finally, since our simulations involve randomness, we run each experiment for a certain configuration of the non-random parameters 10 times and average the results.
The error bars in all graphs denote the minimum and maximum result values across all runs of the experiment.

\subsection{Results}

\subsubsection{Optimal policy for fixed amounts and symmetric/asymmetric demand}

We first simulate a symmetric workload for the channel of 500 transactions on each side with Poisson parameters equal to 3 (on average 1 transaction every 3 seconds), fixed amounts equal to 50, and a shared buffer between the nodes.
The buffering time for all transactions is drawn from a uniform distribution between 0 and a maximum value, and we vary this maximum value across experiments to be 1, 2,..., 10, 20, 30,..., 120 seconds.  

\begin{figure}
    \centering
    \subfigure[Symmetric demand]{
        \includegraphics[width=0.22\textwidth]{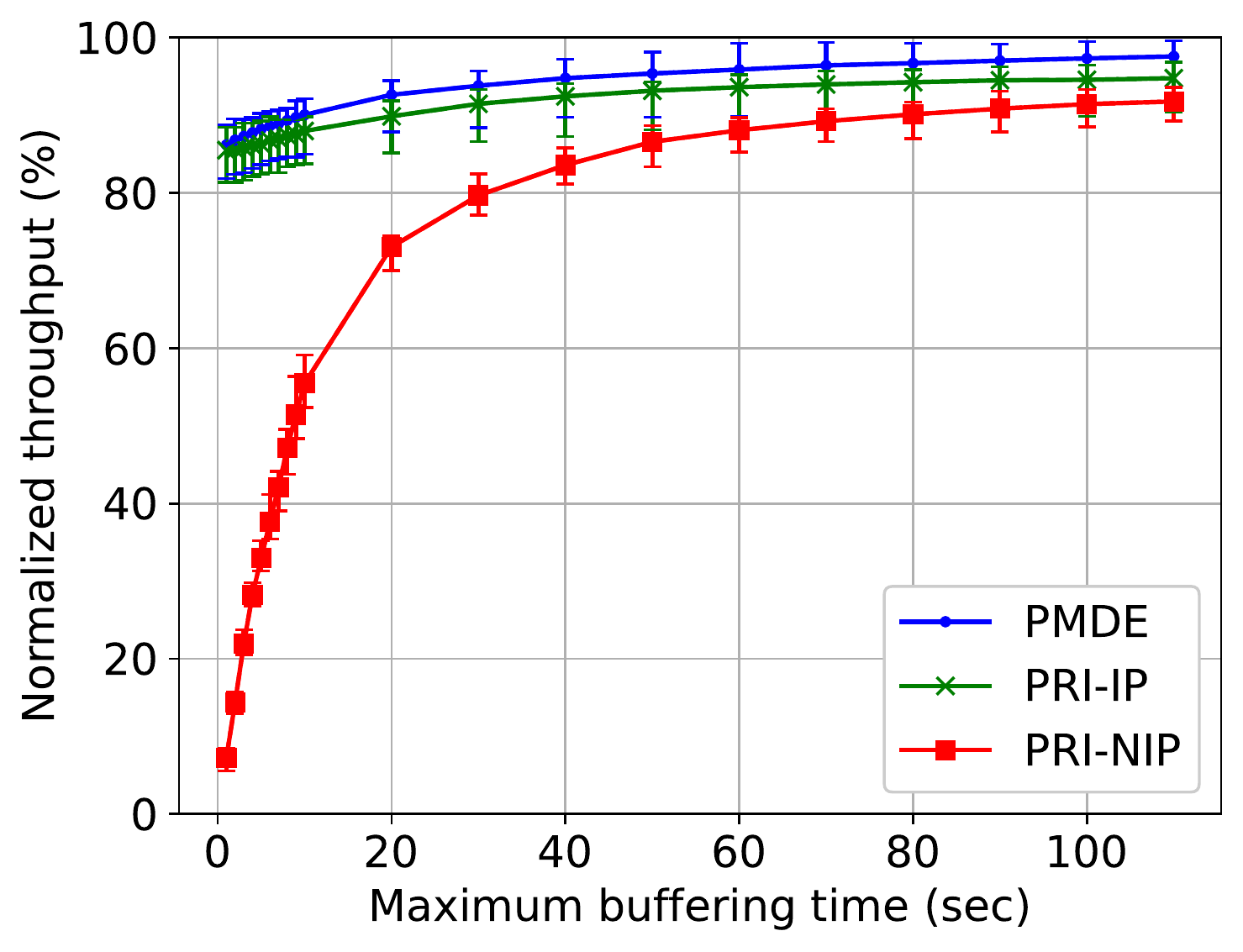}
        \label{fig:results_101_1}
    }
    \subfigure[Asymmetric demand, total throughput]{
        \includegraphics[width=0.22\textwidth]{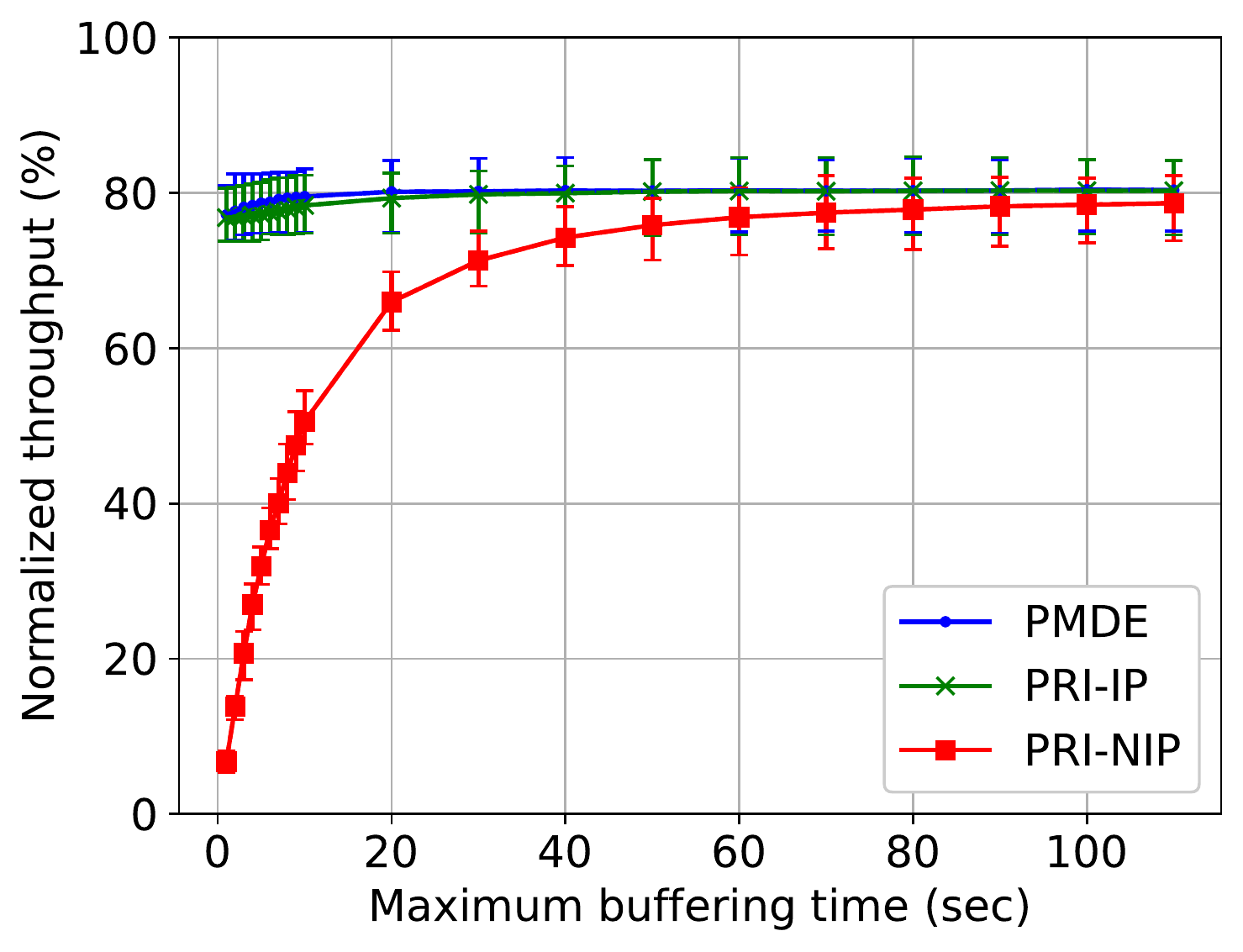}
        \label{fig:results_102_1}
    }
    \subfigure[Asymmetric demand, per node throughput]{
        \includegraphics[width=0.22\textwidth]{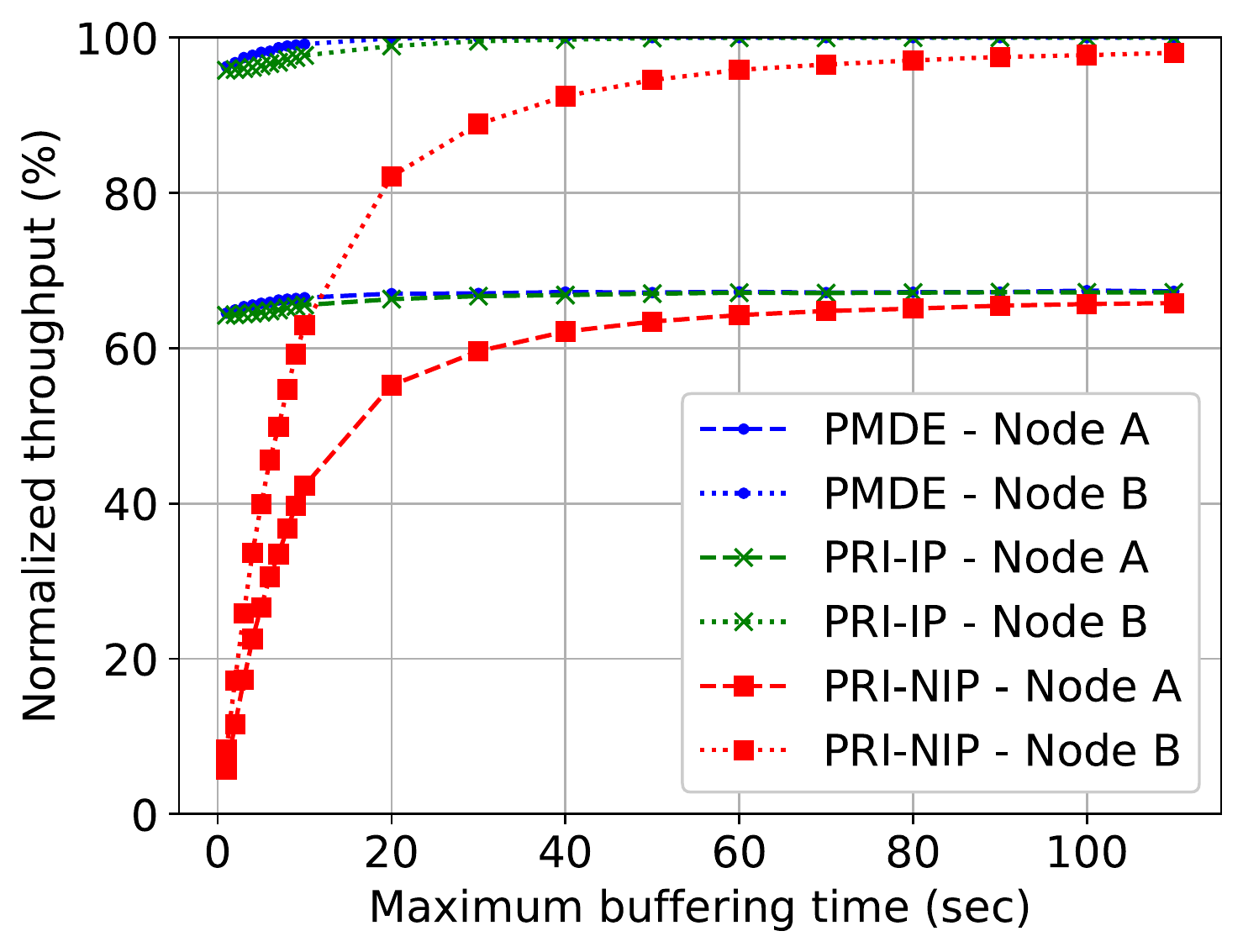}
        \label{fig:results_102_1_nodes}
    }
    \subfigure[Asymmetric demand, Number of sacrificed transactions]{
        \includegraphics[width=0.22\textwidth]{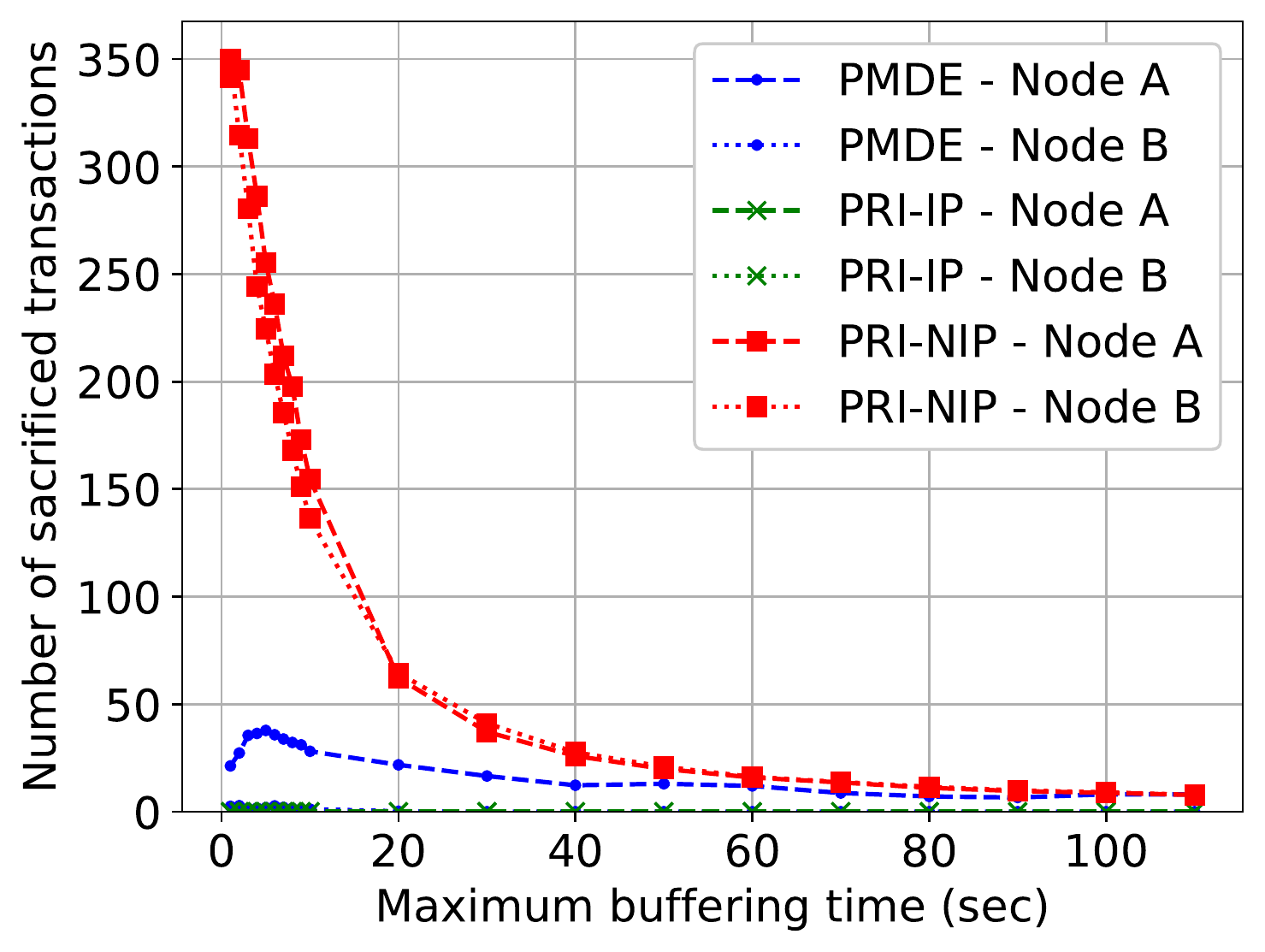}
        \label{fig:results_102_1_sac}
    }

    \caption{Total channel throughput and number of sacrificed transactions as a function of maximum buffering time for different scheduling policies, for oldest-first buffer discipline}
    \label{fig:results_101_102}
\end{figure}

We plot the behavior of the total channel throughput (proportional to the success rate because of fixed amounts) for a single channel for different experiments with increasing maximum buffering time (Figure \ref{fig:results_101_1}).
The figures for the other disciplines are very similar.
Indeed, PMDE performs better than the heuristic PRI policies, as expected.
We also observe for all policies the desired behavior of \textit{increasing throughput with increasing buffering time}.
Moreover, we observe a \textit{diminishing returns behavior}.

We next consider the effects of asymmetry in the payment demand: we modify the setup of the previous section so that now 750 transactions arrive at node $A$ every 2 seconds on average (and 500 every 3 seconds at node $B$).
The results are shown in Figure \ref{fig:results_102_1}.
In this asymmetric demand case,  as expected, the throughput is overall lower compared to the symmetric case, since many transactions from the side with the higher demand do not find enough balance in the channel to become feasible.
Figure \ref{fig:results_102_1_nodes} shows separately the throughput for each of the nodes.
We observe again that buffering is helpful for both nodes, more so for node $B$ though, which was burdened with a smaller load and achieves higher throughput than node $A$.
%still higher though compared to if it did not have a buffer.
It is also interesting that the number of \textit{sacrificed} transactions (i.e. that were feasible on arrival but entered the buffer and were eventually dropped) shown in Figure \ref{fig:results_102_1_sac} is small for PMDE compared to PRI-NIP (and trivially 0 for PRI-IP).

Nevertheless, in both the symmetric and the asymmetric cases, we generally observe what we would expect: that the channel equipped with a buffer (denoted by a non-zero maximum buffering time in the figures) performs at least as good as the channel without a buffer (i.e. with a maximum buffering time equal to 0 in the figures). 

The immediate processing version of PRI leads to slightly better throughput for large buffering times.
The difference between PRI-IP and PRI-NIP is more pronounced for small maximum buffering time values on the horizontal axis, because of the \textit{checkInterval} parameter (set to 3 seconds): for small buffering times, all or most transactions have an allowed time in the buffer of a few seconds, so none, one, or very few chances of being considered every 3 seconds.
The conclusion is that the benefit PRI can reap from holding feasible incoming transactions in the buffer instead of processing them right away is not worth the cost in this case, as processing them immediately leads to higher overall throughput.

\subsubsection{Optimal policy for arbitrary amounts}

We now evaluate our policies on scenarios with symmetric demand when the transaction amounts follow some non-constant distribution.
Specifically, we use a Gaussian distribution of mean 100 and variance 50 (truncated at the channel capacity of 300), a uniform distribution in the interval [0, capacity], and the empirical distribution from the credit card transaction dataset.

We first examine the role of the buffer discipline.
Figure \ref{fig:results_105} shows all the policies for all 5 disciplines for the empirical dataset.
% when the distribution is Gaussian
The figures when using the Gaussian or uniform amounts are similar.
We observe similar results for different buffer disciplines, with PMDE performing best for small and medium maximum buffering times, and PRI-PI performing best for large maximum buffering times.
This is likely due to the fact that PRI-IP offers each transaction multiple chances to be executed (every \textit{checkInterval}), unlike PMDE that offers only one chance.
The higher the maximum buffering time, the more chances transactions get, leading to the higher throughput of PRI.
Since the results are quite similar for different disciplines, in the rest of the figures we adopt the \textit{oldest-first} discipline, which additionally incorporates a notion of First-In-First-Out fairness for transactions.

\begin{figure}[h]
\t\centering
    \subfigure[Oldest first]{
        \includegraphics[width=0.18\textwidth]{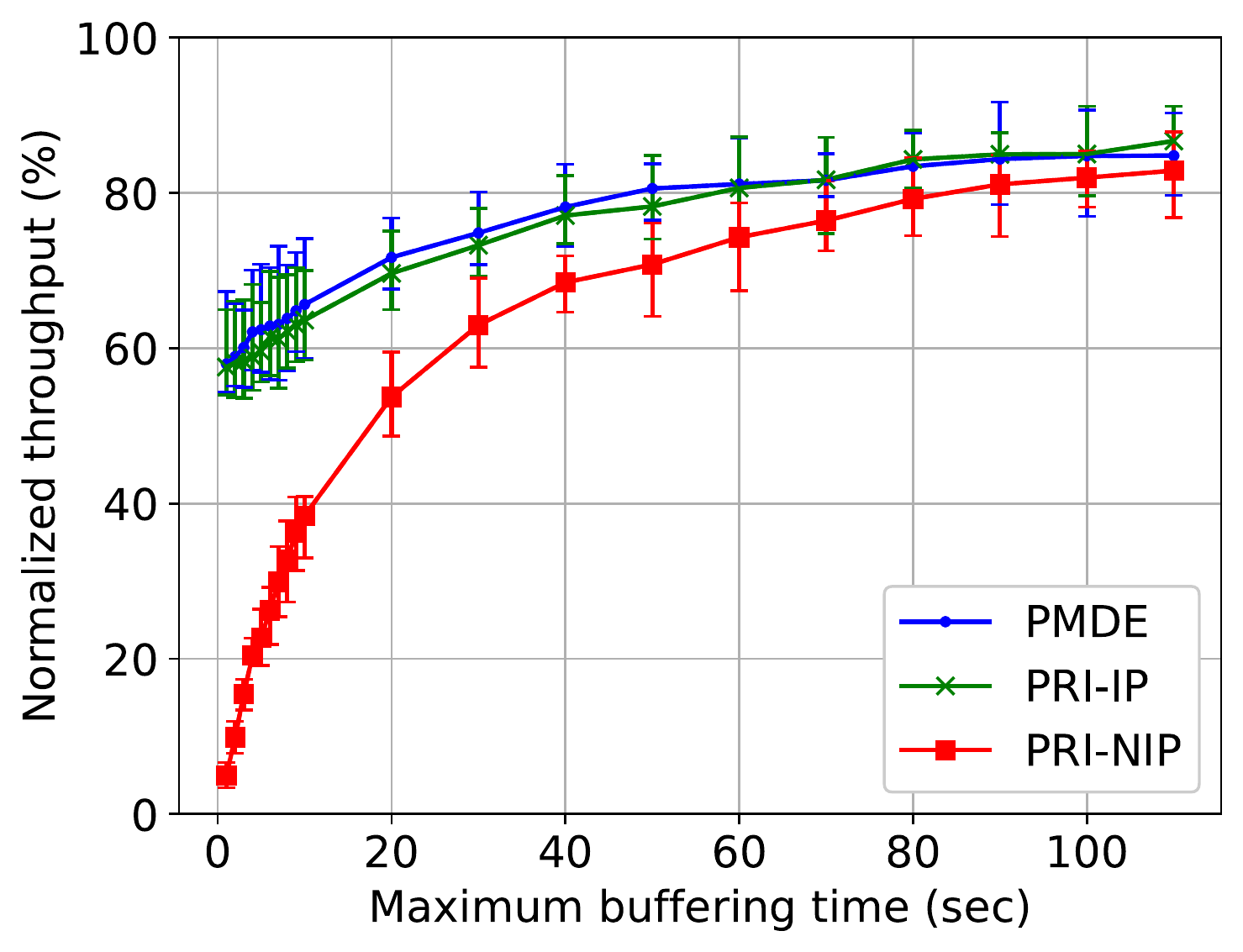}
        \label{fig:results_105_1}
    }
    \subfigure[Youngest first]{
        \includegraphics[width=0.18\textwidth]{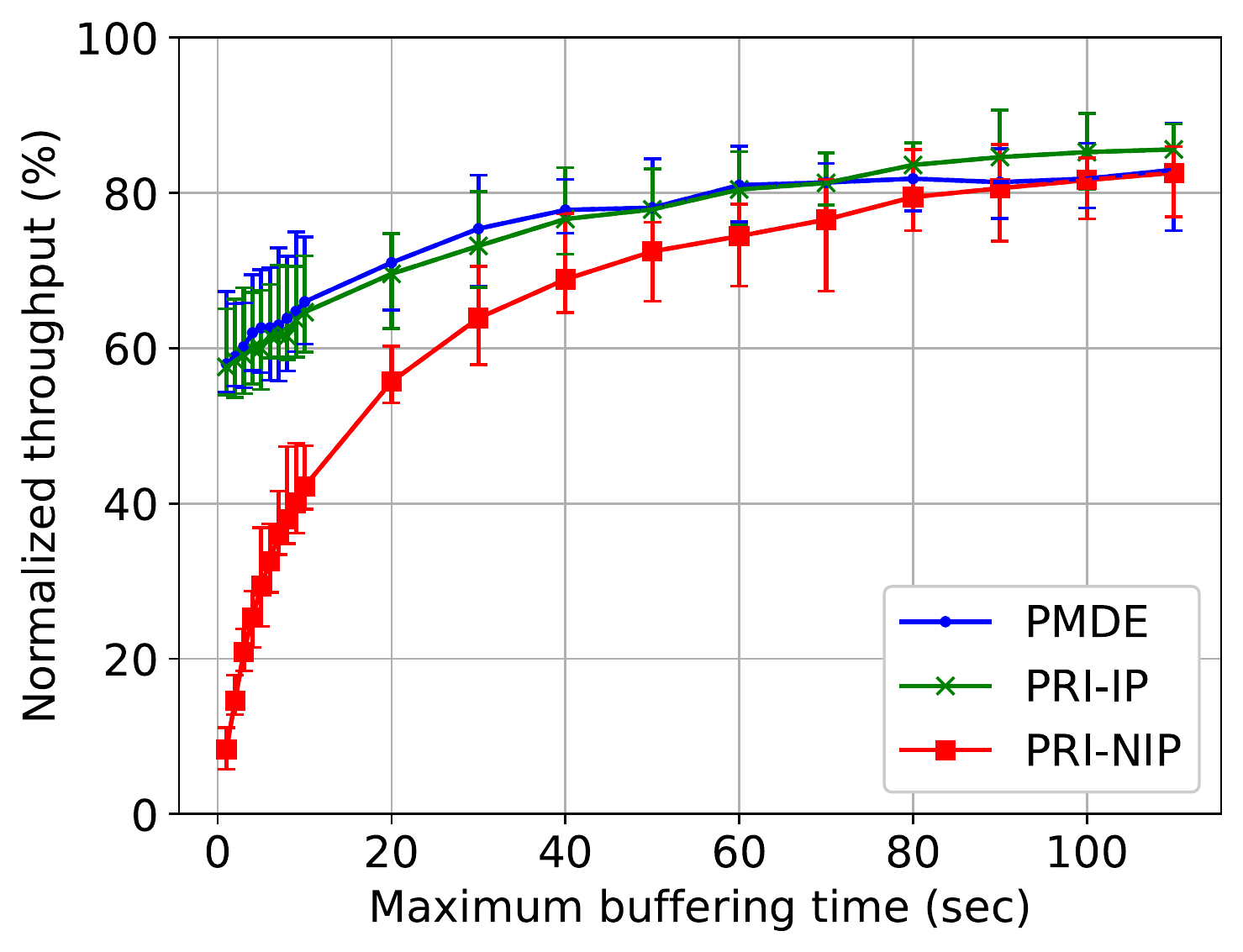}
        \label{fig:results_105_2}
    }
    \subfigure[Closest deadline first]{
        \includegraphics[width=0.18\textwidth]{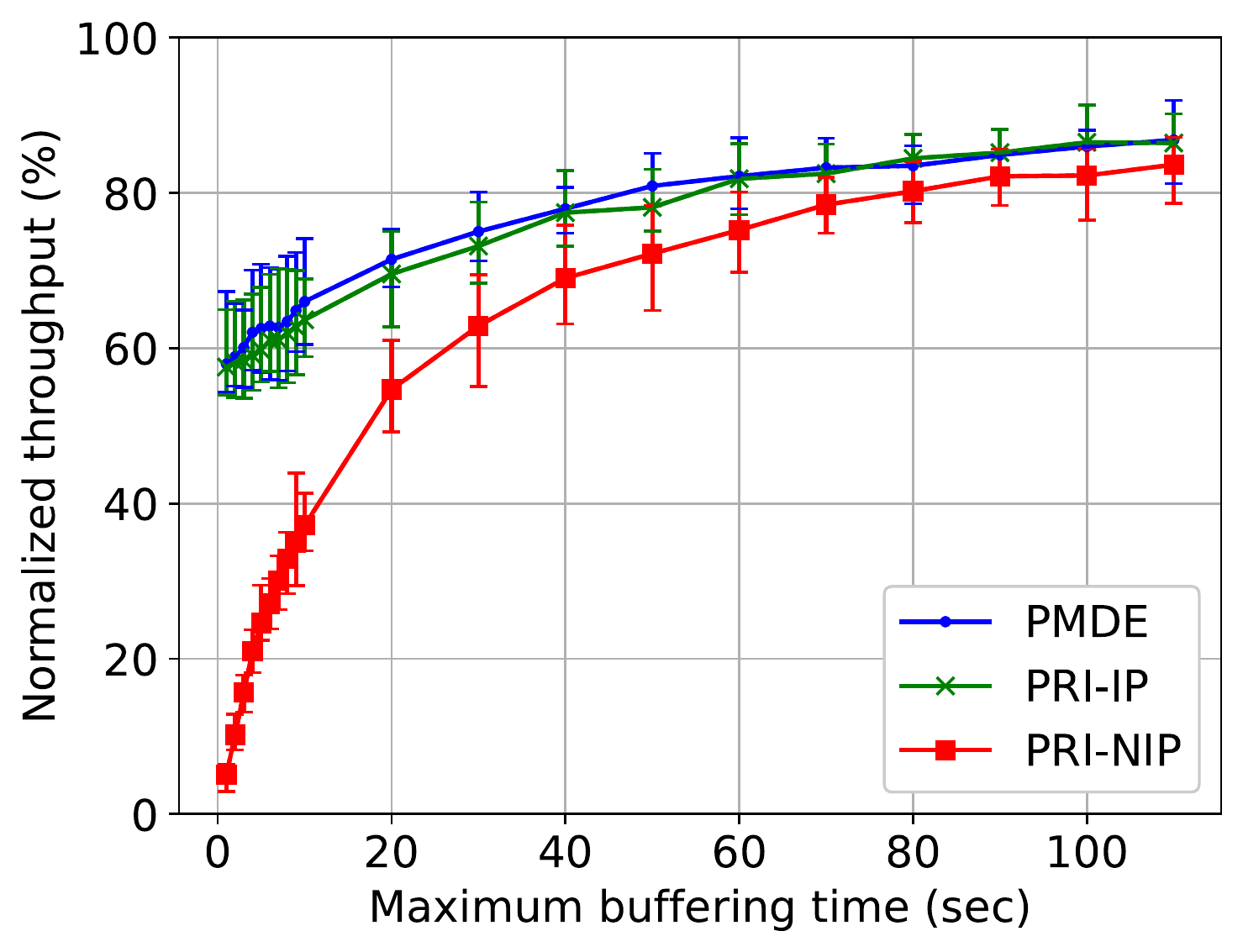}
        \label{fig:results_105_3}
    }
    \subfigure[Largest amount first]{
        \includegraphics[width=0.18\textwidth]{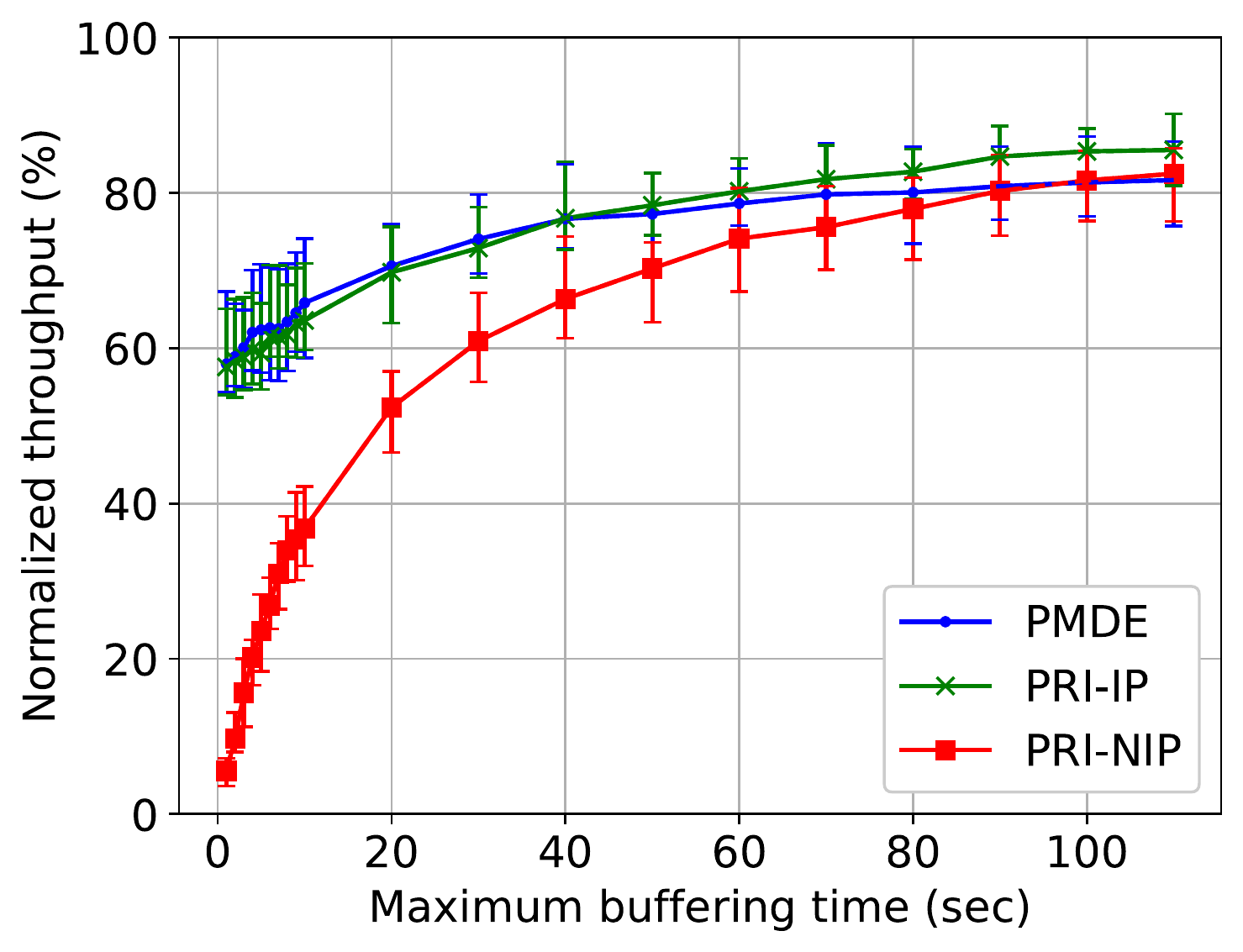}
        \label{fig:results_105_4}
    }
    \subfigure[Smallest amount first]{
        \includegraphics[width=0.18\textwidth]{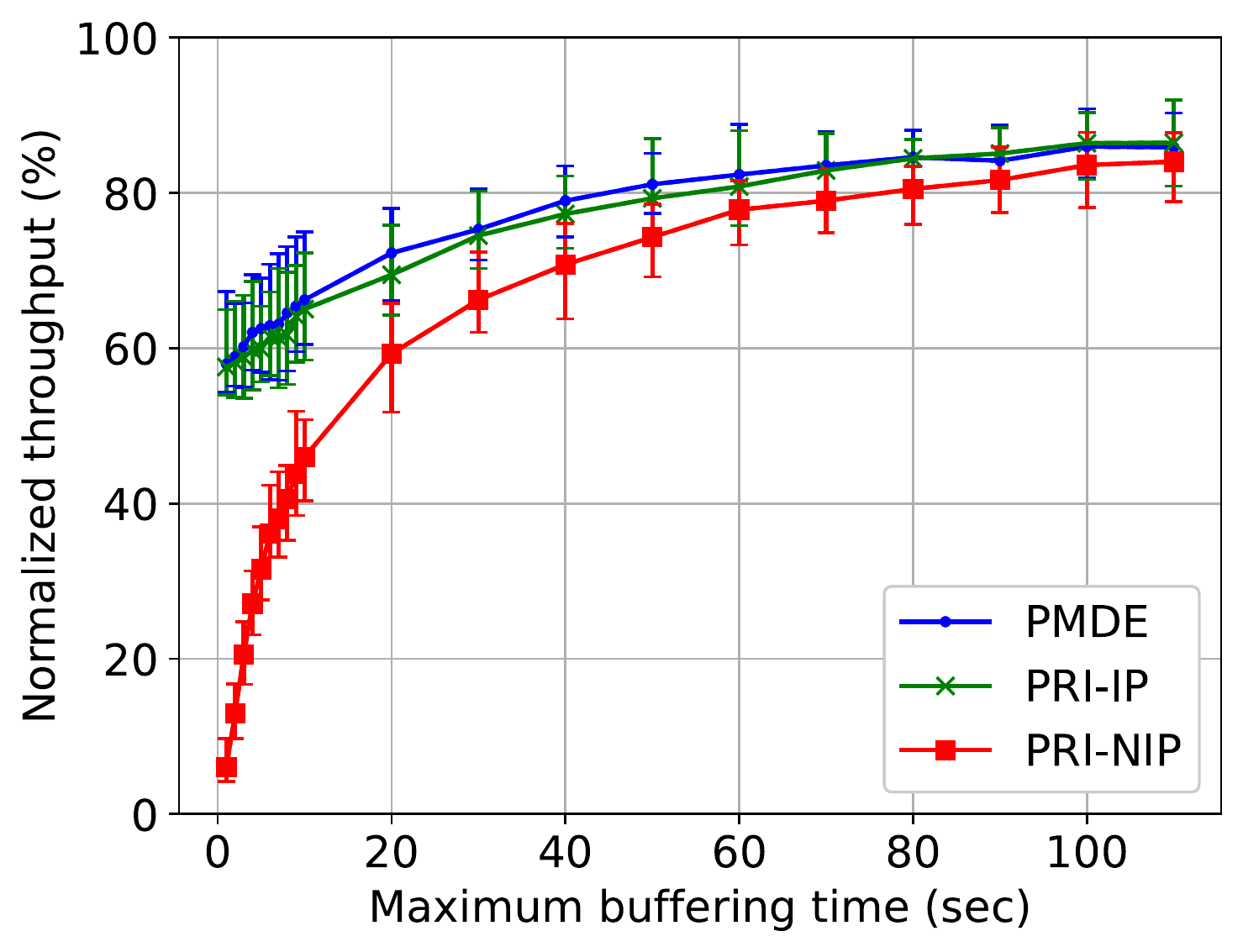}
        \label{fig:results_105_5}
    }

    \caption{Total channel throughput as a function of maximum buffering time for different scheduling policies and buffer disciplines}
    \label{fig:results_105}
\end{figure}

% \begin{figure}[h]
%     \centering
%     \subfigure[Oldest first]{
%         \includegraphics[width=0.3\textwidth]{figures/results_103_1.png}
%         \label{fig:results_105_1}
%     }
%     \subfigure[Youngest first]{
%         \includegraphics[width=0.3\textwidth]{figures/results_103_2.png}
%         \label{fig:results_105_2}
%     }
%     \subfigure[Closest deadline first]{
%         \includegraphics[width=0.3\textwidth]{figures/results_103_3.png}
%         \label{fig:results_105_3}
%     }\\
%     \subfigure[Largest amount first]{
%         \includegraphics[width=0.3\textwidth]{figures/results_103_4.png}
%         \label{fig:results_105_4}
%     }
%     \subfigure[Smallest amount first]{
%         \includegraphics[width=0.3\textwidth]{figures/results_103_5.png}
%         \label{fig:results_105_5}
%     }    

%     \caption{Total channel throughput as a function of maximum buffering time for different scheduling policies and buffer disciplines}
%     \label{fig:results_105}
% \end{figure}

Figure \ref{fig:results_104_105_107} shows the normalized throughput achieved by the three policies under the oldest-first discipline, for different amount distributions.
For Gaussian amounts (Figure \ref{fig:results_103_1}), PMDE outperforms PRI-IP and PRI-NIP.
For uniformly distributed amounts in [0, 300], however (Figure \ref{fig:results_104_1}), we see that for large buffering times PMDE is not as good as the PRI policies.
This is due to the fact that, unlike the Gaussian amounts that were centered around a small value (100), amounts now are more frequently very large, close to the capacity.
As PMDE gives only one chance to transactions to be executed (i.e. on their expiration deadline), while PRI gives them multiple opportunities (i.e. every time the buffer is scanned), very large transactions have a higher probability of being dropped under PMDE than under PRI.
This justification is confirmed by the fact that for smaller Uniform[0, 100] amounts (Figure \ref{fig:results_107_1}), PMDE is indeed the best.
As in practice sending transactions close to the capacity does not constitute good practice and use of a channel, PMDE proves to be the best choice for small- and medium-sized transactions.

\begin{figure}[h]
    \centering
    \subfigure[Gaussian(100, 50)]{
        \includegraphics[width=0.23\textwidth]{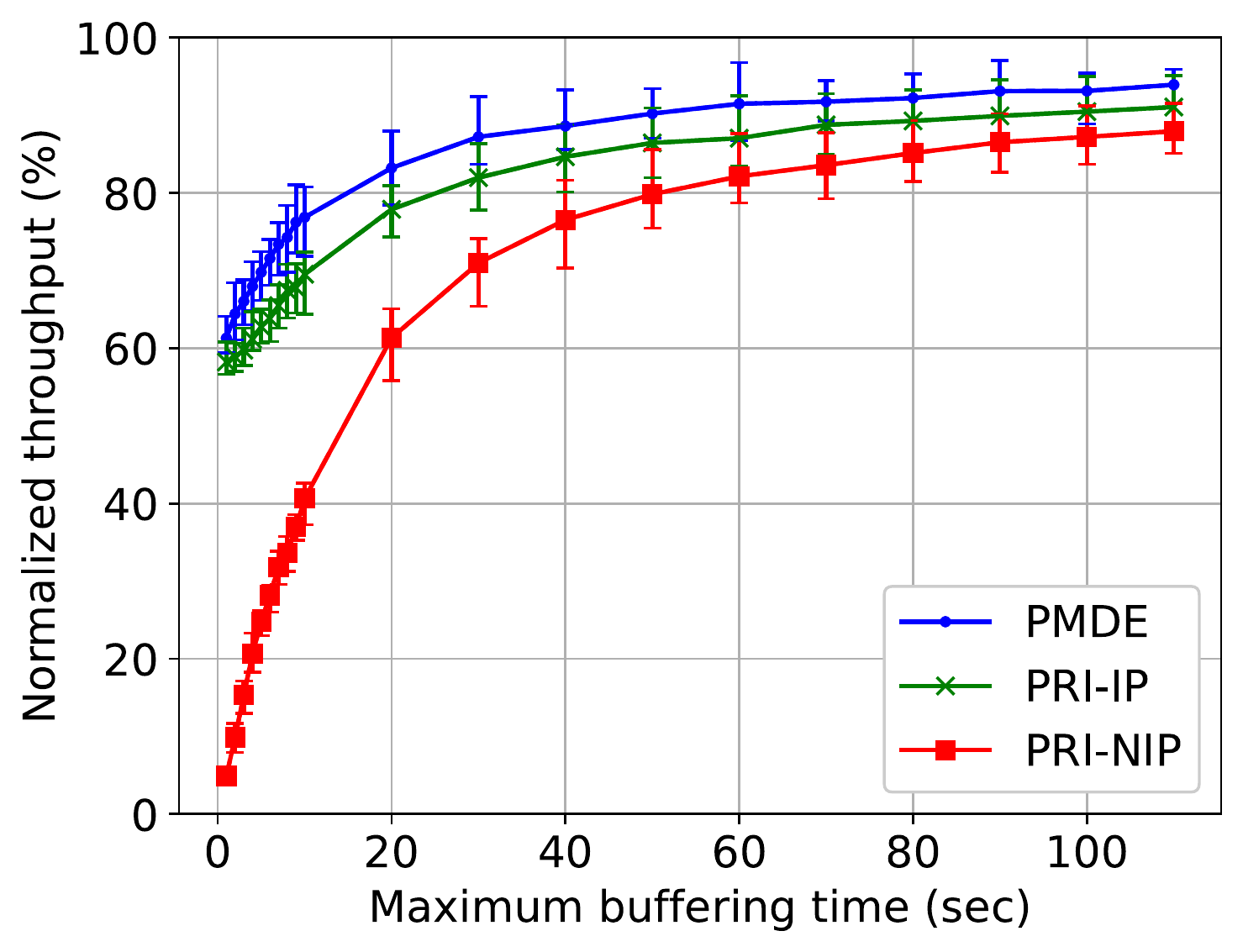}
        \label{fig:results_103_1}
    }
    \subfigure[Uniform(0, 300)]{
        \includegraphics[width=0.23\textwidth]{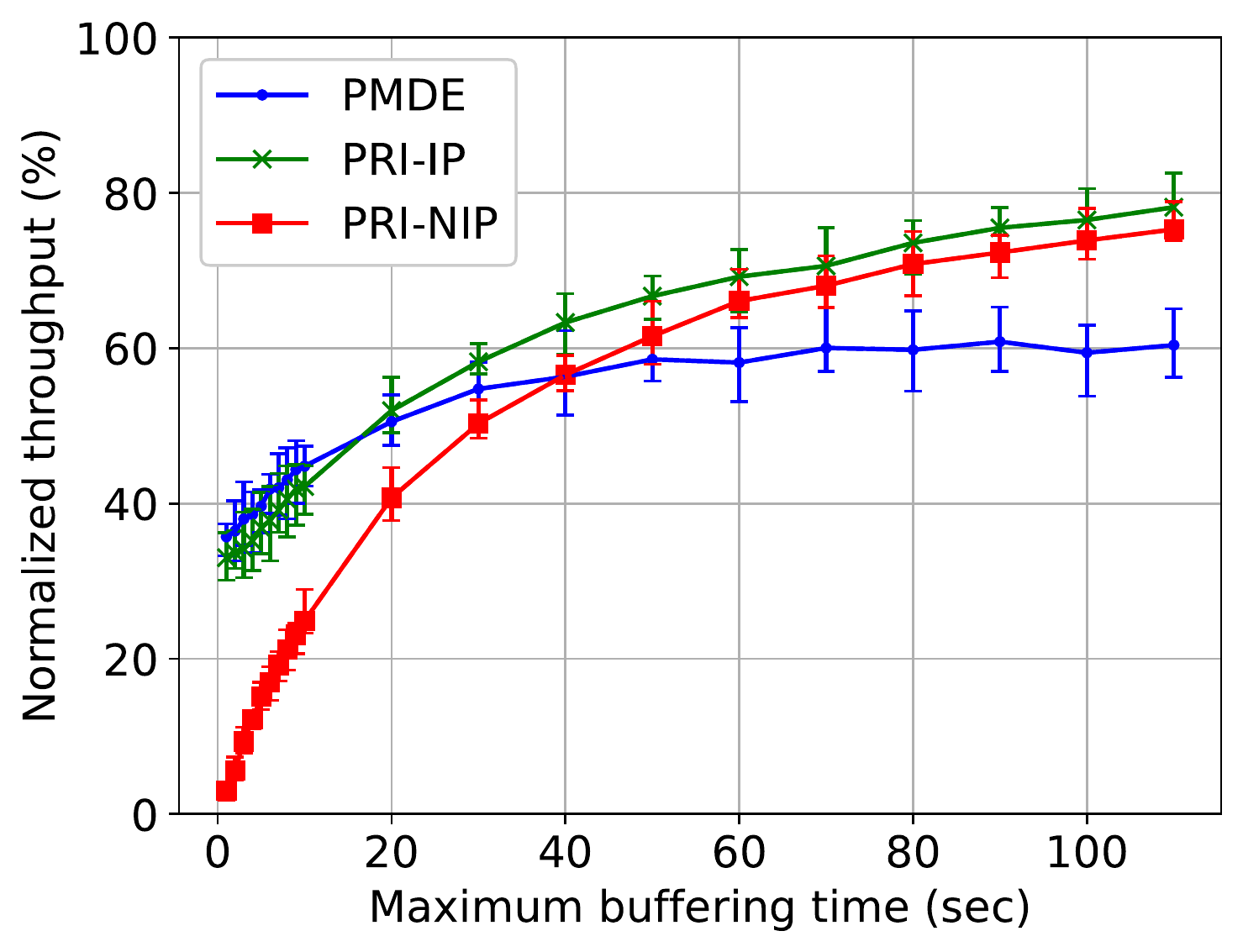}
        \label{fig:results_104_1}
    }
    \subfigure[Uniform(0, 100)]{
        \includegraphics[width=0.23\textwidth]{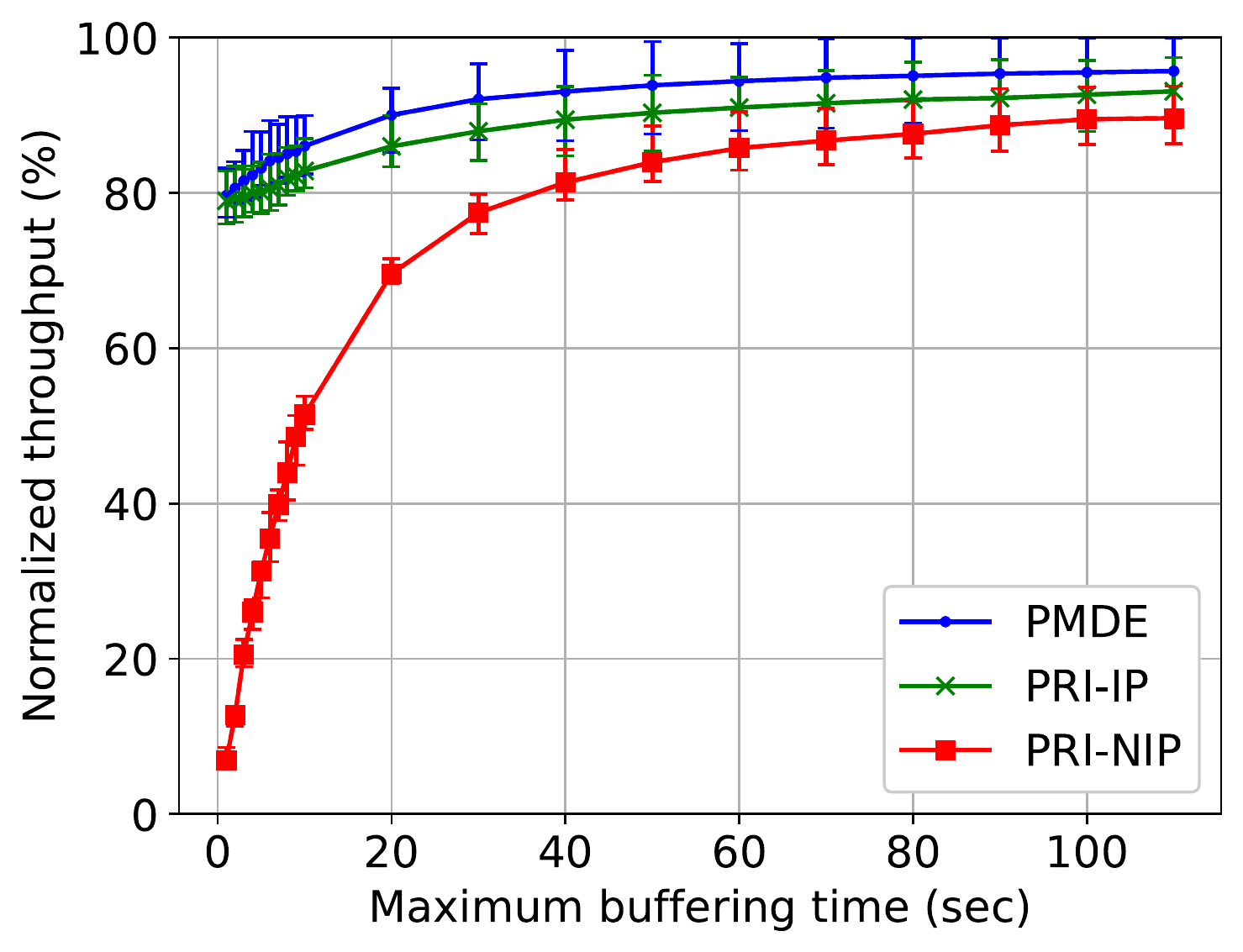}
        \label{fig:results_107_1}
    }
    \subfigure[Empirical]{
        \includegraphics[width=0.23\textwidth]{figures/results_105_1.pdf}
        \label{fig:results_105_1-b}
    }

    \caption{Total channel throughput as a function of maximum buffering time for different scheduling policies and transaction amount distributions}
    \label{fig:results_104_105_107}
\end{figure}

\subsubsection{The importance of privacy and collaboration in scheduling}

We now study a different question: how important is it for \textit{both} nodes to have a buffer, and if they do, to share the contents with the other node (a node might have a buffer and not want to share its contents for privacy reasons).
As mentioned earlier, for PMDE in particular this concern is not applicable, as the only information shared is essentially the expiring transaction(s), which would be revealed anyway at the time of their execution.
For PRI though, a policy prioritizing the oldest transaction in the entire buffer versus in one direction only might have better performance, and provide an incentive to nodes to share their buffers.
% The privacy concern becomes even more relevant in a network with multiple nodes and multiple channels between them. 
% Buffers now might contain transactions with multiple next hops. 
% The fact that for a channel A-B, A and B would find the optimal policy over the contents of the entire buffers of A and B constitutes a privacy leakage for e.g. C, as in A's buffer there exists transactions towards C. So in this case a privacy problem exists. However, it can be solved if we assume that the buffers are per-channel. Is this optimal though? Maybe not, but we are now deviating a lot from the single-channel focus of the paper.

\begin{wrapfigure}{L}{0.45\textwidth}
% \begin{figure}[h]
    \centering
    \subfigure[PMDE]{
        \includegraphics[width=0.2\textwidth]{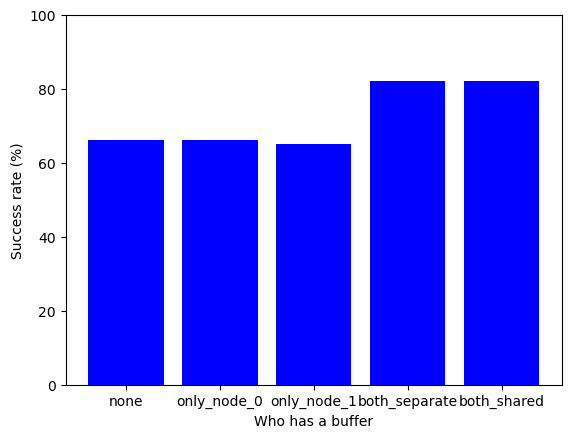}
        \label{fig:results_106_PMDE_sr}
    }
    \subfigure[PMDE]{
        \includegraphics[width=0.2\textwidth]{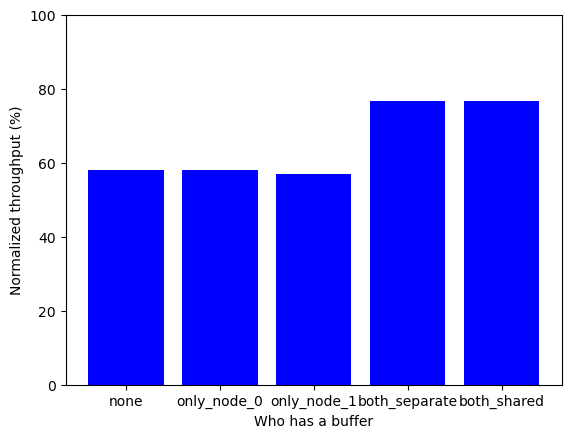}
        \label{fig:results_106_PMDE_nthr}
    }\\
    \subfigure[PRI-IP]{
        \includegraphics[width=0.2\textwidth]{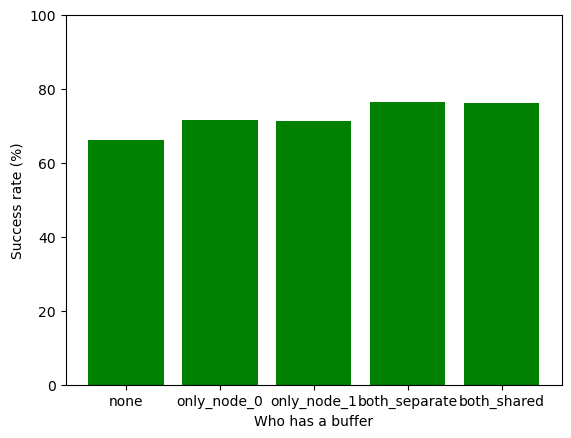}
        \label{fig:results_106_PRI-IP_sr}
    }
    \subfigure[PRI-IP]{
        \includegraphics[width=0.2\textwidth]{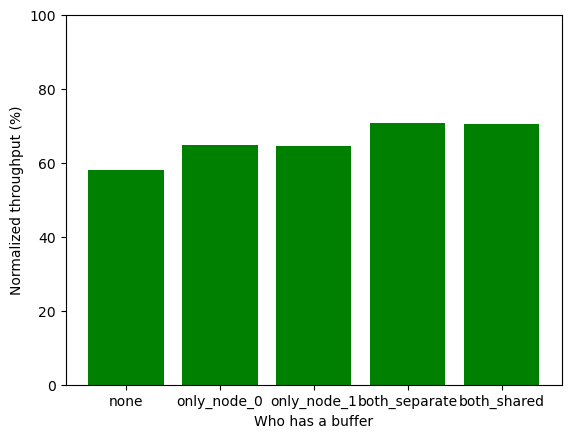}
        \label{fig:results_106_PRI-IP_nthr}
    }
    \subfigure[PRI-NIP]{
        \includegraphics[width=0.2\textwidth]{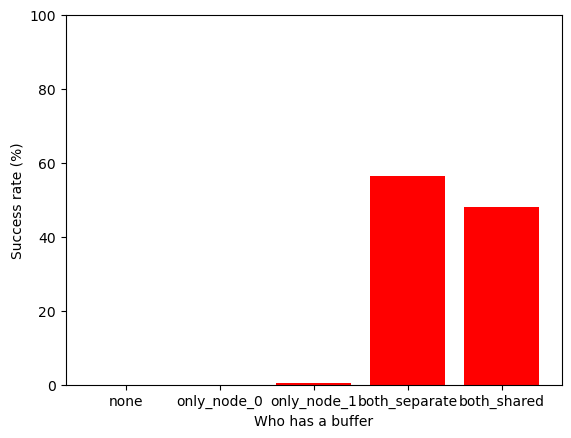}
        \label{fig:results_106_PRI-NIP_sr}
    }
    \subfigure[PRI-NIP]{
        \includegraphics[width=0.2\textwidth]{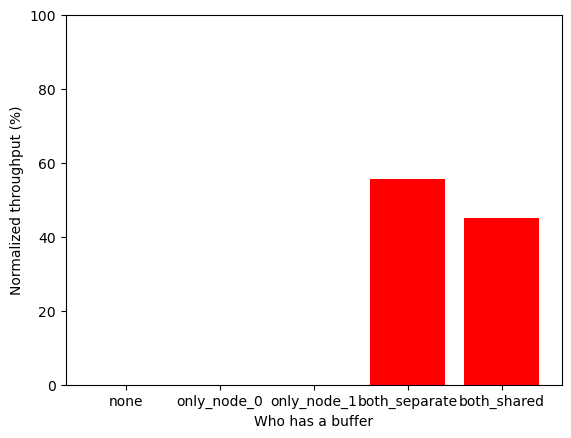}
        \label{fig:results_106_PRI-NIP_nthr}
    }
    % \caption{Total normalized throughput for different scheduling policies and node buffering capabilities}
    % \label{fig:results_106_nthr}
    \caption{Success rate and normalized throughput for different scheduling policies and node buffering capabilities}
    \label{fig:results_106}
% \end{figure}
\end{wrapfigure}

We evaluate a scenario with symmetric demand of 500 transactions from each side every 3 seconds on average, with Gaussian amounts as before, and buffering times uniform in [0, 5] seconds.
We evaluate all policies with the oldest-first discipline for all combinations of buffer capabilities at nodes: none, only one or the other, both but without shared knowledge, and both with shared knowledge.
% The results for the success rate are shown in Figure \ref{fig:results_106_sr} and for the normalized throughput in Figure \ref{fig:results_106_nthr}.
The results for the success rate and the normalized throughput are shown in Figure \ref{fig:results_106}.

We observe that all policies perform better when both nodes have buffers as opposed to one or both of them not having.
Non-immediate processing trivially leads to almost 0 performance when at least one node does not have a buffer because all transactions of this node are dropped (by being redirected to a non-existent buffer), and thus neither can the other node execute any but few transactions because its side gets depleted after executing the first few.
In conclusion, PRI-NIP makes sense only when both nodes have buffers.
We also observe similar performance in PMDE and PRI-IP for separate and shared buffers, which suggests that nodes can apply these policies while keeping their buffer contents private without missing out on performance.
(In PRI-NIP, they actually even miss out on performance by sharing).

\begin{figure}[h]
    \centering
    % \subfigure[Buffers: none,\protect\\Amounts: Gaussian(100, 50)]{
    \subfigure[Buffers: none, Amounts: Gaussian(100, 50)]{
        \includegraphics[width=0.23\textwidth]{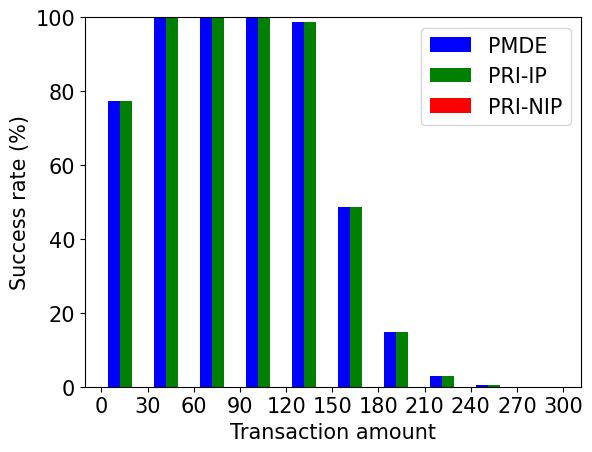}
        \label{fig:results_111_1}
    }
    % \subfigure[Buffers: both nodes,\protect\\Amounts: Gaussian(100, 50)]{
    \subfigure[Buffers: both nodes, Amounts: Gaussian(100, 50)]{
        \includegraphics[width=0.23\textwidth]{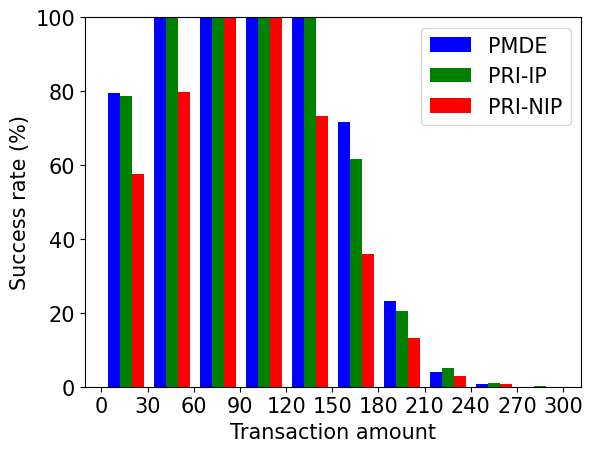}
        \label{fig:results_112_1}
    }
    % \subfigure[Buffers: none,\protect\\Amounts: Empirical]{
    \subfigure[Buffers: none, Amounts: Empirical]{
        \includegraphics[width=0.23\textwidth]{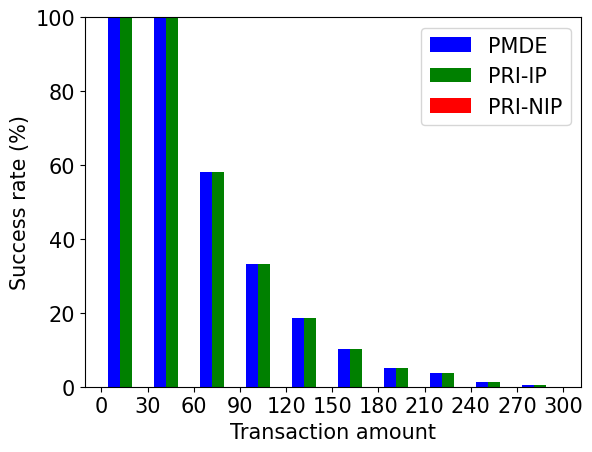}
        \label{fig:results_113_1}
    }
    % \subfigure[Buffers: both nodes,\protect\\Amounts: Empirical]{
    \subfigure[Buffers: both nodes, Amounts: Empirical]{
        \includegraphics[width=0.23\textwidth]{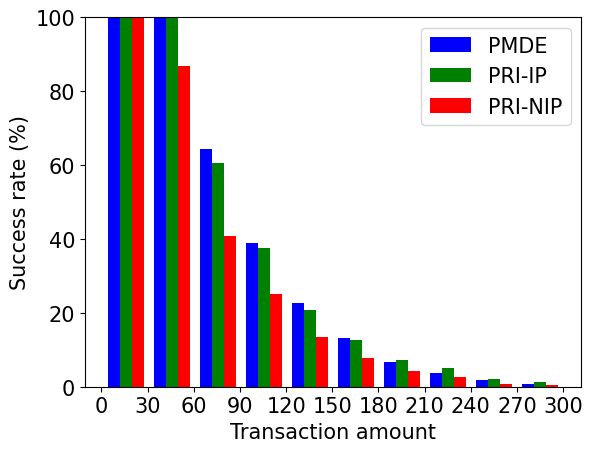}
        \label{fig:results_114_1}
    }
    \caption{Success rate as a function of transaction amount for different scheduling policies}
    \label{fig:results_111_112_113_114}
\end{figure}

\subsubsection{Benefits from buffering as a function of the transaction amount}
\label{sec:amount-graphs}

We now study how the existence of buffers affects the throughput of transactions of different amounts.
We run one experiment with a specific configuration: initial balances 0 and 300, Gaussian(100, 50) transaction amounts, and constant deadlines for all transactions equal to 5 seconds. 
We repeat the experiment 10 times and average the results. 
We partition the transaction amounts in intervals and plot the success rate of transactions in each interval.
We do the same for amounts from the empirical distribution.
The result for \textit{oldest-first} buffer discipline are shown in Figure \ref{fig:results_111_112_113_114} (results for other disciplines are similar).

\begin{wrapfigure}{R}{0.3\textwidth}
    \centering
    \includegraphics[width=0.3\textwidth]{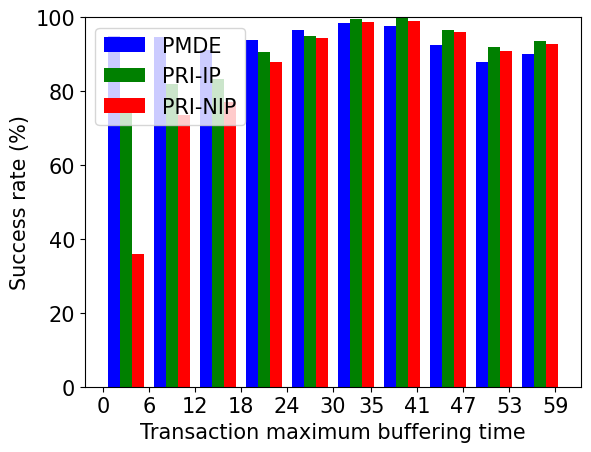}
    \caption{Success rate as a function of transaction maximum buffering time for different scheduling policies}
    \label{fig:results_115}
\end{wrapfigure}

By comparing the graphs where the nodes do not have buffers versus when they both have a shared buffer, we observe that it is the transactions of larger amounts that are actually benefiting from buffering.
The reason is that smaller transactions are more likely to be feasible and clear on their arrival even without a buffer, while larger ones are likely infeasible on arrival.
The zero success rates of PRI-NIP when there are no buffers are trivially due to its design.
We observe similar success rates for PMDE and PRI-IP when there are no buffers, and PMDE being slightly better than PRI-IP when there are buffers, except for possibly a few very large amounts (the latter is for the same reason why PRI-IP is better for large amounts in Figure \ref{fig:results_104_1}).
This insight is important from a user experience perspective: a PCN node, depending on the sizes of the transactions it serves, can decide whether it is worthwhile to use PMDE for higher success rates but have to wait longer for each transaction (till its deadline expiration), or use some more immediate policy like PRI with potentially lower success rate but faster clearing.

\subsubsection{Benefits from buffering as a function of the transaction deadline}
Similarly as in section \ref{sec:amount-graphs}, in this section we study whether transactions with longer versus shorter initial buffering times tend to benefit from the existence of the buffer the most.
We run one experiment with a specific configuration: initial balances 0 and 300, constant transaction amounts equal to 50, and uniform deadlines from 0 to 60 seconds. 
We repeat the experiment 10 times and average the results. 
We partition the buffering times in intervals and plot the success rate of transactions in each interval.
The result for \textit{oldest-first} buffer discipline are shown in Figure \ref{fig:results_115} (results for other disciplines are similar).
We observe that for PMDE there is no differentiation among transactions with different buffering times, as PMDE processes all transactions on their deadline expiration, regardless of when that occurs.
For the PRI policies though, large buffering times (e.g., more than 11 seconds) are generally better, as they allow for more opportunities for the transaction to be considered for processing (recall the buffer is scanned every 3 seconds).
The user experience insight from this experiment is that if a node decides to use PMDE for some reason related for example to the transaction amounts, the deadline values do not matter in terms of the success rate.
On the other hand, if PRI is used, the node should know that this might be disadvantaging transactions with short buffering times.

% % Add different plot to show sacrificed transactions.

\section{Extensions to a network setting}
\label{sec:extensions}

In order to extend the PCN throughput maximization problem to an entire network $G=(V,E)$ with a node set $V$ and an edge set $E$, we need to redefine our objective and deal with other factors in our decision making that arise when we have non-direct payments.
The objective in a network setting would be to maximize the total over all pairs of different nodes $S = \sum_{(i, j) \in E} S_{ij}$, where $S_{ij}$ is the throughput of the channel between nodes $i$ and $j$.
The control in a network setting is the policy each node follows in each of its channels.

\subsection{The complete graph case}

The single channel model we have described so far can be immediately extended to model a PCN that is a complete graph. 
In a complete graph, if we assume that transactions are always routed along the shortest path in hop count, all transactions will succeed or fail without needing to take a multihop route. 
Then, all the channels are independent of each other, and choosing the policies for each node that maximize the total network throughput can be decomposed to choosing a policy for each channel separately.

\subsection{The star graph case}

Let us now consider a star graph: all the payments between peripheral nodes have to pass through the central node. 
In this case, the shortest path between a pair of nodes $i$ and $j$ is unique: from $i$ to the central node and from the central node to $j$. 
Moreover, the paths between any pairs of nodes $(i_1, j_1)$, $(i_2, j_2)$, with $i_1,j_1,i_2,j_2$ distinct, are non-overlapping, so the share of the total throughput that corresponds to these paths is the sum of the throughput of each path.
However, for paths where for example $j_1 = j_2$, the policy the central node applies might also depend on whether an arriving payment arrives from $i_1$ or arrives from $i_2$. 
The central node does have this knowledge and may use it to prioritize transactions of $i_1$ vs of $i_2$, or to follow an entirely different policy for transactions arriving from $i_1$ than for transactions arriving from $i_2$.
This shows one more factor that a scheduling policy has to consider in a multihop network apart from the amount and deadline of each transaction: the origin and destination of the transaction.
In the single channel, this information was only used as a binary direction of the payment.

\subsection{The general case}

Unlike the star graph, in a general multihop network there might be multiple shortest\footnote{Shortest can be defined in terms of hops, as the cheapest in terms of fees, or a combination thereof, as is done in Lightning.} paths between a pair of nodes.
Thus, another decision to be made for each transaction is the routing decision: which of the alternative paths to use. 
There are two cases:
\begin{enumerate}
    \item nodes use source routing for payments: the origin node determines the entire path for the payment until it reaches the destination. In this case, intermediate nodes do not make any routing decision; they just forward the payment to the next predetermined hop.
    \item nodes use distributed routing: the node at each hop determines the next one. In this case, the control decision at each node are both the scheduling and the routing decisions.
\end{enumerate}

Deadlines are also more complicated to reason about in a network setting: in a multihop network, there are two possibilities.
Transactions either have an end-to-end deadline by which they have to have been processed or dropped, or have a per-hop deadline by which they have to have been forwarded to the next hop or dropped.
The per-hop deadlines could be chosen from the original sender, however choosing them in a ``right'' way to maximize the throughput is not straightforward.

In conclusion, when seeking generality, a holistic approach to both routing and scheduling is needed.
We believe that stochastic modeling and optimization techniques can be a useful tool towards making optimal decisions based on the details of the network- and channel-level interactions
In addition, as the joint problems become more complex and do not lend themselves to analytical solutions, reinforcement learning can assist in utilizing the insights given by the data trail of the network's operation to empirically derive optimal operational parameters and policies. 
We leave the exploration of these directions to future work.

\section{Related work}
\label{sec:related-work}

Most of the research at a network level in PCNs has focused on the routing problem for multihop transactions.
A channel rebalancing technique is studied in \cite{Pickhardt2019}. In \cite{Tang2020}, privacy-utility tradeoffs in payment channel are explored, in terms of the benefit in success rates nodes can have by revealing noisy versions of thir channel balances.
In \cite{Wang2019}, payments are categorized into ``elephant'' and ``mice'' payments and a different routing approach is followed for each category.

The problem of taking optimal scheduling decisions for arriving payments in the channels of a PCN has not been studied extensively in the literature.
The most relevant work to ours is probably \cite{Sivaraman2020}, which introduces a routing approach for nodes in a PCN that aims to maximize throughput, via packetization of transactions and a transport protocol for congestion control in the different nodes of the network. 
The paper assumes the existence of queues at the different channels, with transaction units queued-up whenever the channel lacks the funds to process them immediately, and a one-bit congestion signal from the routers that helps throttle the admitted demand so that congestion and channel depletion are avoided. 
The paper's focus is on routing, and the scheduling policies used for the queues are heuristically chosen. 
In contrast, we propose that queueing can be beneficial to the overall throughput even if the transaction is feasible on its arrival and opt for a more formal to come up with optimal policies.
Another important difference is that \cite{Sivaraman2020} uses a fluid model for the incoming transaction demand, while we model the demand as distinct incoming transactions arriving as a marked point process and base our policy decisions on the particular characteristics of the specific transactions.

Another interesting relevant work is \cite{Varma2019}, which focuses on throughput-maximizing routing policies: designing a path from the sender to the recipient of each transaction so that the network throughput is maximized and the use of on-chain rebalancing is minimized. 
It proposes dynamic MaxWeight-based routing policies, uses a discrete time stochastic model and models the channel as a double-sided queue, like the ones usually used in ride-hailing systems.
Our model in contrast is a continuous time one, focuses more on scheduling rather than routing, and avoids certain limitations arising from the double-sided queue assumption by modeling the channel state using two separate queues, one for each side.

Finally, \cite{ACD} considers a Payment Service Provider (PSP), a node that can establish multiple channels and wants to profit from relaying others' payments in the network. 
The goal is to define a strategy of the PSP that will determine which of the incoming transactions to process in order to maximize profit from fees while minimizing the capital locked in channels. 
The paper shows that even a simple variant of the scheduling problem is NP-hard, and proposes a polynomial approximation algorithm. 
However, the assumption throughout the paper is that transactions have to be executed or dropped as soon as and in the order in which they arrive, and this differentiates the problem compared to our case.

\section{Conclusion}
\label{sec:conclusion}

% In summary, our paper is the first to formally treat the optimal scheduling problem in a PCN with buffering capabilities.

In this paper, we studied the transaction scheduling problem in PCNs. 
We defined the PMDE policy and proved its optimality for constant arriving amount distributions.
We also defined a heuristic extension of PMDE as well as heuristic policies PRI for arbitrary amount distributions, and studied in detail the policies via experiments in our simulator.
This work opens the way for further rigorous results in problems of networking nature arising in PCNs.
In the future, we hope to see research on joint routing and scheduling mechanisms that will be able to push the potential of PCNs to their physical limits and make them a scalable and reliable solution for financial applications and beyond.

% \begin{acks}
% \end{acks}

%% The next two lines define the bibliography style to be used, and the bibliography file.
\bibliographystyle{ACM-Reference-Format}
\bibliography{references}


\begin{thebibliography}{27}

%%% the Web address.

\ifx \showCODEN    \undefined \def \showCODEN     #1{\unskip}     \fi
\ifx \showDOI      \undefined \def \showDOI       #1{#1}\fi
\ifx \showISBNx    \undefined \def \showISBNx     #1{\unskip}     \fi
\ifx \showISBNxiii \undefined \def \showISBNxiii  #1{\unskip}     \fi
\ifx \showISSN     \undefined \def \showISSN      #1{\unskip}     \fi
\ifx \showLCCN     \undefined \def \showLCCN      #1{\unskip}     \fi
\ifx \shownote     \undefined \def \shownote      #1{#1}          \fi
\ifx \showarticletitle \undefined \def \showarticletitle #1{#1}   \fi
\ifx \showURL      \undefined \def \showURL       {\relax}        \fi
% invisible to TeX
\providecommand\bibfield[2]{#2}
\providecommand\bibinfo[2]{#2}
\providecommand\natexlab[1]{#1}
\providecommand\showeprint[2][]{arXiv:#2}

\bibitem[\protect\citeauthoryear{??}{Rai}{[n.d.]}]%
        {RaidenNetwork}
 \bibinfo{year}{[n.d.]}\natexlab{}.
\newblock \bibinfo{title}{{Raiden Network}}.
\newblock
\newblock
\urldef\tempurl%
\url{https://raiden.network/101.html}
\showURL{%
\tempurl}


\bibitem[\protect\citeauthoryear{{Aaron van Wirdum}}{{Aaron van
  Wirdum}}{[n.d.]}]%
        {cltv}
\bibfield{author}{\bibinfo{person}{{Aaron van Wirdum}}.}
  \bibinfo{year}{[n.d.]}\natexlab{}.
\newblock \bibinfo{title}{{Understanding the Lightning Network, Part 1:
  Building a Bidirectional Bitcoin Payment Channel}}.
\newblock
\newblock
\urldef\tempurl%
\url{https://bitcoinmagazine.com/articles/understanding-the-lightning-network-part-building-a-bidirectional-payment-channel-1464710791}
\showURL{%
\tempurl}


\bibitem[\protect\citeauthoryear{Avarikioti, Wang, and Wattenhofer}{Avarikioti
  et~al\mbox{.}}{2018}]%
        {ACD}
\bibfield{author}{\bibinfo{person}{Georgia Avarikioti}, \bibinfo{person}{Yuyi
  Wang}, {and} \bibinfo{person}{Roger Wattenhofer}.}
  \bibinfo{year}{2018}\natexlab{}.
\newblock \showarticletitle{Algorithmic channel design}. In
  \bibinfo{booktitle}{\emph{29th International Symposium on Algorithms and
  Computation (ISAAC 2018)}} \emph{(\bibinfo{series}{Leibniz International
  Proceedings in Informatics (LIPIcs)}, Vol.~\bibinfo{volume}{123})},
  \bibfield{editor}{\bibinfo{person}{Wen-Lian Hsu}, \bibinfo{person}{Der-Tsai
  Lee}, {and} \bibinfo{person}{Chung-Shou Liao}} (Eds.).
  \bibinfo{publisher}{Schloss Dagstuhl--Leibniz-Zentrum fuer Informatik},
  \bibinfo{address}{Dagstuhl, Germany}, \bibinfo{pages}{16:1--16:12}.
\newblock
\showISBNx{978-3-95977-094-1}
\showISSN{1868-8969}
\urldef\tempurl%
\url{https://doi.org/10.4230/LIPIcs.ISAAC.2018.16}
\showDOI{\tempurl}


\bibitem[\protect\citeauthoryear{Bagaria, Kannan, Tse, Fanti, and
  Viswanath}{Bagaria et~al\mbox{.}}{2019}]%
        {prism}
\bibfield{author}{\bibinfo{person}{Vivek Bagaria}, \bibinfo{person}{Sreeram
  Kannan}, \bibinfo{person}{David Tse}, \bibinfo{person}{Giulia Fanti}, {and}
  \bibinfo{person}{Pramod Viswanath}.} \bibinfo{year}{2019}\natexlab{}.
\newblock \showarticletitle{Prism: Deconstructing the Blockchain to Approach
  Physical Limits}. In \bibinfo{booktitle}{\emph{Proceedings of the 2019 ACM
  SIGSAC Conference on Computer and Communications Security}} (London, United
  Kingdom) \emph{(\bibinfo{series}{CCS '19})}. \bibinfo{publisher}{Association
  for Computing Machinery}, \bibinfo{address}{New York, NY, USA},
  \bibinfo{pages}{585–602}.
\newblock
\showISBNx{9781450367479}
\urldef\tempurl%
\url{https://doi.org/10.1145/3319535.3363213}
\showDOI{\tempurl}


\bibitem[\protect\citeauthoryear{Croman, Decker, Eyal, Gencer, Juels, Kosba,
  Miller, Saxena, Shi, G{\"u}n~Sirer, Song, and Wattenhofer}{Croman
  et~al\mbox{.}}{2016}]%
        {Cromanetal2016}
\bibfield{author}{\bibinfo{person}{Kyle Croman}, \bibinfo{person}{Christian
  Decker}, \bibinfo{person}{Ittay Eyal}, \bibinfo{person}{Adem~Efe Gencer},
  \bibinfo{person}{Ari Juels}, \bibinfo{person}{Ahmed Kosba},
  \bibinfo{person}{Andrew Miller}, \bibinfo{person}{Prateek Saxena},
  \bibinfo{person}{Elaine Shi}, \bibinfo{person}{Emin G{\"u}n~Sirer},
  \bibinfo{person}{Dawn Song}, {and} \bibinfo{person}{Roger Wattenhofer}.}
  \bibinfo{year}{2016}\natexlab{}.
\newblock \showarticletitle{On Scaling Decentralized Blockchains}. In
  \bibinfo{booktitle}{\emph{Financial Cryptography and Data Security}},
  \bibfield{editor}{\bibinfo{person}{Jeremy Clark}, \bibinfo{person}{Sarah
  Meiklejohn}, \bibinfo{person}{Peter~Y.A. Ryan}, \bibinfo{person}{Dan
  Wallach}, \bibinfo{person}{Michael Brenner}, {and} \bibinfo{person}{Kurt
  Rohloff}} (Eds.). \bibinfo{publisher}{Springer Berlin Heidelberg},
  \bibinfo{address}{Berlin, Heidelberg}, \bibinfo{pages}{106--125}.
\newblock
\showISBNx{978-3-662-53357-4}


\bibitem[\protect\citeauthoryear{Dong, Liang, Li, and Liu}{Dong
  et~al\mbox{.}}{2018}]%
        {celer}
\bibfield{author}{\bibinfo{person}{Mo Dong}, \bibinfo{person}{Qingkai Liang},
  \bibinfo{person}{Xiaozhou Li}, {and} \bibinfo{person}{Junda Liu}.}
  \bibinfo{year}{2018}\natexlab{}.
\newblock \showarticletitle{{Celer Network}: Bring Internet Scale to Every
  Blockchain}.
\newblock \bibinfo{journal}{\emph{CoRR}}  \bibinfo{volume}{abs/1810.00037}
  (\bibinfo{year}{2018}).
\newblock
\showeprint[arxiv]{1810.00037}
\urldef\tempurl%
\url{http://arxiv.org/abs/1810.00037}
\showURL{%
\tempurl}


\bibitem[\protect\citeauthoryear{Grunspan, Leh{\'{e}}ricy, and
  P{\'{e}}rez{-}Marco}{Grunspan et~al\mbox{.}}{2020}]%
        {Grunspan2018b}
\bibfield{author}{\bibinfo{person}{Cyril Grunspan}, \bibinfo{person}{Gabriel
  Leh{\'{e}}ricy}, {and} \bibinfo{person}{Ricardo P{\'{e}}rez{-}Marco}.}
  \bibinfo{year}{2020}\natexlab{}.
\newblock \showarticletitle{Ant Routing scalability for the {Lightning
  Network}}.
\newblock \bibinfo{journal}{\emph{CoRR}}  \bibinfo{volume}{abs/2002.01374}
  (\bibinfo{year}{2020}).
\newblock
\showeprint[arxiv]{2002.01374}
\urldef\tempurl%
\url{https://arxiv.org/abs/2002.01374}
\showURL{%
\tempurl}


\bibitem[\protect\citeauthoryear{Gudgeon, Moreno-Sanchez, Roos, McCorry, and
  Gervais}{Gudgeon et~al\mbox{.}}{2020}]%
        {Gudgeon2020a}
\bibfield{author}{\bibinfo{person}{Lewis Gudgeon}, \bibinfo{person}{Pedro
  Moreno-Sanchez}, \bibinfo{person}{Stefanie Roos}, \bibinfo{person}{Patrick
  McCorry}, {and} \bibinfo{person}{Arthur Gervais}.}
  \bibinfo{year}{2020}\natexlab{}.
\newblock \showarticletitle{{SoK}: Layer-Two Blockchain Protocols}. In
  \bibinfo{booktitle}{\emph{Financial Cryptography and Data Security}},
  \bibfield{editor}{\bibinfo{person}{Joseph Bonneau} {and}
  \bibinfo{person}{Nadia Heninger}} (Eds.). \bibinfo{publisher}{Springer
  International Publishing}, \bibinfo{address}{Cham},
  \bibinfo{pages}{201--226}.
\newblock
\showISBNx{978-3-030-51280-4}


\bibitem[\protect\citeauthoryear{Hoenisch and Weber}{Hoenisch and
  Weber}{2018}]%
        {Hoenisch2018}
\bibfield{author}{\bibinfo{person}{Philipp Hoenisch} {and}
  \bibinfo{person}{Ingo Weber}.} \bibinfo{year}{2018}\natexlab{}.
\newblock \showarticletitle{{AODV}--Based Routing for Payment Channel
  Networks}. In \bibinfo{booktitle}{\emph{Blockchain -- ICBC 2018}},
  \bibfield{editor}{\bibinfo{person}{Shiping Chen}, \bibinfo{person}{Harry
  Wang}, {and} \bibinfo{person}{Liang-Jie Zhang}} (Eds.).
  \bibinfo{publisher}{Springer International Publishing},
  \bibinfo{address}{Cham}, \bibinfo{pages}{107--124}.
\newblock
\showISBNx{978-3-319-94478-4}


\bibitem[\protect\citeauthoryear{L{\"{u}}nsdorf and Scherfke}{L{\"{u}}nsdorf
  and Scherfke}{[n.d.]}]%
        {simpy}
\bibfield{author}{\bibinfo{person}{Ontje L{\"{u}}nsdorf} {and}
  \bibinfo{person}{Stefan Scherfke}.} \bibinfo{year}{[n.d.]}\natexlab{}.
\newblock \bibinfo{title}{{SimPy}}.
\newblock
\newblock
\urldef\tempurl%
\url{https://simpy.readthedocs.io}
\showURL{%
\tempurl}


\bibitem[\protect\citeauthoryear{{Machine Learning Group - ULB}}{{Machine
  Learning Group - ULB}}{[n.d.]}]%
        {kaggle-dataset}
\bibfield{author}{\bibinfo{person}{{Machine Learning Group - ULB}}.}
  \bibinfo{year}{[n.d.]}\natexlab{}.
\newblock \bibinfo{title}{{Credit Card Fraud Detection - Anonymized credit card
  transactions labeled as fraudulent or genuine}}.
\newblock
\newblock
\urldef\tempurl%
\url{https://www.kaggle.com/mlg-ulb/creditcardfraud}
\showURL{%
\tempurl}


\bibitem[\protect\citeauthoryear{Malavolta, Moreno{-}Sanchez, Kate, and
  Maffei}{Malavolta et~al\mbox{.}}{2017}]%
        {Silentwhispers}
\bibfield{author}{\bibinfo{person}{Giulio Malavolta}, \bibinfo{person}{Pedro
  Moreno{-}Sanchez}, \bibinfo{person}{Aniket Kate}, {and}
  \bibinfo{person}{Matteo Maffei}.} \bibinfo{year}{2017}\natexlab{}.
\newblock \showarticletitle{{SilentWhispers}: Enforcing Security and Privacy in
  Decentralized Credit Networks}. In \bibinfo{booktitle}{\emph{24th Annual
  Network and Distributed System Security Symposium, {NDSS} 2017, San Diego,
  California, USA, February 26 - March 1, 2017}}. \bibinfo{publisher}{The
  Internet Society}.
\newblock
\urldef\tempurl%
\url{https://www.ndss-symposium.org/ndss2017/ndss-2017-programme/silentwhispers-enforcing-security-and-privacy-decentralized-credit-networks/}
\showURL{%
\tempurl}


\bibitem[\protect\citeauthoryear{Nakamoto}{Nakamoto}{2008}]%
        {Nakamoto2008}
\bibfield{author}{\bibinfo{person}{Satoshi Nakamoto}.}
  \bibinfo{year}{2008}\natexlab{}.
\newblock \showarticletitle{{Bitcoin: A Peer-to-Peer Electronic Cash System}}.
\newblock  (\bibinfo{year}{2008}).
\newblock
\urldef\tempurl%
\url{https://bitcoin.org/bitcoin.pdf}
\showURL{%
\tempurl}


\bibitem[\protect\citeauthoryear{Osuntokun and Fromknecht}{Osuntokun and
  Fromknecht}{2018}]%
        {AMP}
\bibfield{author}{\bibinfo{person}{Olaoluwa Osuntokun} {and}
  \bibinfo{person}{Conner Fromknecht}.} \bibinfo{year}{2018}\natexlab{}.
\newblock \bibinfo{title}{{Atomic Multipath Payments}}.
\newblock
\newblock
\urldef\tempurl%
\url{https://lists.linuxfoundation.org/pipermail/lightning-dev/2018-February/000993.html}
\showURL{%
\tempurl}


\bibitem[\protect\citeauthoryear{Papadis, Borst, Walid, Grissa, and
  Tassiulas}{Papadis et~al\mbox{.}}{2018}]%
        {Papadis2018}
\bibfield{author}{\bibinfo{person}{Nikolaos Papadis}, \bibinfo{person}{Sem
  Borst}, \bibinfo{person}{Anwar Walid}, \bibinfo{person}{Mohamed Grissa},
  {and} \bibinfo{person}{Leandros Tassiulas}.} \bibinfo{year}{2018}\natexlab{}.
\newblock \showarticletitle{{Stochastic Models and Wide-Area Network
  Measurements for Blockchain Design and Analysis}}. In
  \bibinfo{booktitle}{\emph{IEEE INFOCOM 2018 - IEEE Conference on Computer
  Communications}}. \bibinfo{publisher}{IEEE}, \bibinfo{pages}{2546--2554}.
\newblock
\showISBNx{978-1-5386-4128-6}
\showISSN{0743166X}
\urldef\tempurl%
\url{https://doi.org/10.1109/INFOCOM.2018.8485982}
\showDOI{\tempurl}


\bibitem[\protect\citeauthoryear{Papadis and Tassiulas}{Papadis and
  Tassiulas}{2020}]%
        {Papadis2020}
\bibfield{author}{\bibinfo{person}{Nikolaos Papadis} {and}
  \bibinfo{person}{Leandros Tassiulas}.} \bibinfo{year}{2020}\natexlab{}.
\newblock \showarticletitle{{Blockchain-Based Payment Channel Networks:
  Challenges and Recent Advances}}.
\newblock \bibinfo{journal}{\emph{IEEE Access}}  \bibinfo{volume}{8}
  (\bibinfo{year}{2020}), \bibinfo{pages}{227596--227609}.
\newblock
\showISSN{2169-3536}
\urldef\tempurl%
\url{https://doi.org/10.1109/ACCESS.2020.3046020}
\showDOI{\tempurl}


\bibitem[\protect\citeauthoryear{{Pickhardt} and {Nowostawski}}{{Pickhardt} and
  {Nowostawski}}{2020}]%
        {Pickhardt2019}
\bibfield{author}{\bibinfo{person}{Rene {Pickhardt}} {and}
  \bibinfo{person}{Mariusz {Nowostawski}}.} \bibinfo{year}{2020}\natexlab{}.
\newblock \showarticletitle{Imbalance measure and proactive channel rebalancing
  algorithm for the {Lightning Network}}. In \bibinfo{booktitle}{\emph{2020
  IEEE International Conference on Blockchain and Cryptocurrency (ICBC)}}.
  \bibinfo{pages}{1--5}.
\newblock
\urldef\tempurl%
\url{https://doi.org/10.1109/ICBC48266.2020.9169456}
\showDOI{\tempurl}


\bibitem[\protect\citeauthoryear{Poon and Dryja}{Poon and Dryja}{2016}]%
        {Poon2016}
\bibfield{author}{\bibinfo{person}{Joseph Poon} {and} \bibinfo{person}{Thaddeus
  Dryja}.} \bibinfo{year}{2016}\natexlab{}.
\newblock \bibinfo{title}{{The Bitcoin Lightning Network: scalable off-chain
  instant payments}}.
\newblock
\newblock
\urldef\tempurl%
\url{https://lightning.network/lightning-network-paper.pdf}
\showURL{%
\tempurl}


\bibitem[\protect\citeauthoryear{Prihodko, Zhigulin, Sahno, Ostrovskiy, and
  Osuntokun}{Prihodko et~al\mbox{.}}{2016}]%
        {Prihodko2016}
\bibfield{author}{\bibinfo{person}{Pavel Prihodko}, \bibinfo{person}{Slava
  Zhigulin}, \bibinfo{person}{Mykola Sahno}, \bibinfo{person}{Aleksei
  Ostrovskiy}, {and} \bibinfo{person}{Olaoluwa Osuntokun}.}
  \bibinfo{year}{2016}\natexlab{}.
\newblock \showarticletitle{Flare: An approach to routing in {Lightning
  Network}}.
\newblock \bibinfo{journal}{\emph{Whitepaper}} (\bibinfo{year}{2016}).
\newblock


\bibitem[\protect\citeauthoryear{Rohrer, La{\ss}, and Tschorsch}{Rohrer
  et~al\mbox{.}}{2017}]%
        {rohrer2017}
\bibfield{author}{\bibinfo{person}{Elias Rohrer},
  \bibinfo{person}{Jann-Frederik La{\ss}}, {and} \bibinfo{person}{Florian
  Tschorsch}.} \bibinfo{year}{2017}\natexlab{}.
\newblock \showarticletitle{Towards a Concurrent and Distributed Route
  Selection for Payment Channel Networks}. In \bibinfo{booktitle}{\emph{Data
  Privacy Management, Cryptocurrencies and Blockchain Technology}},
  \bibfield{editor}{\bibinfo{person}{Joaquin Garcia-Alfaro},
  \bibinfo{person}{Guillermo Navarro-Arribas}, \bibinfo{person}{Hannes
  Hartenstein}, {and} \bibinfo{person}{Jordi Herrera-Joancomart{\'i}}} (Eds.).
  \bibinfo{publisher}{Springer International Publishing},
  \bibinfo{address}{Cham}, \bibinfo{pages}{411--419}.
\newblock
\showISBNx{978-3-319-67816-0}


\bibitem[\protect\citeauthoryear{Roos, Moreno{-}Sanchez, Kate, and
  Goldberg}{Roos et~al\mbox{.}}{2018}]%
        {Speedymurmurs}
\bibfield{author}{\bibinfo{person}{Stefanie Roos}, \bibinfo{person}{Pedro
  Moreno{-}Sanchez}, \bibinfo{person}{Aniket Kate}, {and} \bibinfo{person}{Ian
  Goldberg}.} \bibinfo{year}{2018}\natexlab{}.
\newblock \showarticletitle{Settling Payments Fast and Private: Efficient
  Decentralized Routing for Path-Based Transactions}. In
  \bibinfo{booktitle}{\emph{25th Annual Network and Distributed System Security
  Symposium, {NDSS} 2018, San Diego, California, USA, February 18-21, 2018}}.
  \bibinfo{publisher}{The Internet Society}.
\newblock
\urldef\tempurl%
\url{http://wp.internetsociety.org/ndss/wp-content/uploads/sites/25/2018/02/ndss2018{\_}09-3{\_}Roos{\_}paper.pdf}
\showURL{%
\tempurl}


\bibitem[\protect\citeauthoryear{Sivaraman, Venkatakrishnan, Ruan, Negi, Yang,
  Mittal, Fanti, and Alizadeh}{Sivaraman et~al\mbox{.}}{2020}]%
        {Sivaraman2020}
\bibfield{author}{\bibinfo{person}{Vibhaalakshmi Sivaraman},
  \bibinfo{person}{Shaileshh~Bojja Venkatakrishnan}, \bibinfo{person}{Kathleen
  Ruan}, \bibinfo{person}{Parimarjan Negi}, \bibinfo{person}{Lei Yang},
  \bibinfo{person}{Radhika Mittal}, \bibinfo{person}{Giulia Fanti}, {and}
  \bibinfo{person}{Mohammad Alizadeh}.} \bibinfo{year}{2020}\natexlab{}.
\newblock \showarticletitle{High Throughput Cryptocurrency Routing in Payment
  Channel Networks}. In \bibinfo{booktitle}{\emph{17th {USENIX} Symposium on
  Networked Systems Design and Implementation ({NSDI} 20)}}.
  \bibinfo{publisher}{{USENIX} Association}, \bibinfo{address}{Santa Clara,
  CA}, \bibinfo{pages}{777--796}.
\newblock
\showISBNx{978-1-939133-13-7}
\urldef\tempurl%
\url{https://www.usenix.org/conference/nsdi20/presentation/sivaraman}
\showURL{%
\tempurl}


\bibitem[\protect\citeauthoryear{Tang, Wang, Fanti, and Oh}{Tang
  et~al\mbox{.}}{2020}]%
        {Tang2020}
\bibfield{author}{\bibinfo{person}{Weizhao Tang}, \bibinfo{person}{Weina Wang},
  \bibinfo{person}{Giulia Fanti}, {and} \bibinfo{person}{Sewoong Oh}.}
  \bibinfo{year}{2020}\natexlab{}.
\newblock \showarticletitle{Privacy-Utility Tradeoffs in Routing Cryptocurrency
  over Payment Channel Networks}.
\newblock \bibinfo{journal}{\emph{Proc. ACM Meas. Anal. Comput. Syst.}}
  \bibinfo{volume}{4}, \bibinfo{number}{2}, Article \bibinfo{articleno}{29}
  (\bibinfo{date}{June} \bibinfo{year}{2020}), \bibinfo{numpages}{39}~pages.
\newblock
\urldef\tempurl%
\url{https://doi.org/10.1145/3392147}
\showDOI{\tempurl}


\bibitem[\protect\citeauthoryear{{Tassiulas} and {Ephremides}}{{Tassiulas} and
  {Ephremides}}{1993}]%
        {tinfo}
\bibfield{author}{\bibinfo{person}{Leandros {Tassiulas}} {and}
  \bibinfo{person}{Anthony {Ephremides}}.} \bibinfo{year}{1993}\natexlab{}.
\newblock \showarticletitle{Dynamic server allocation to parallel queues with
  randomly varying connectivity}.
\newblock \bibinfo{journal}{\emph{IEEE Transactions on Information Theory}}
  \bibinfo{volume}{39}, \bibinfo{number}{2} (\bibinfo{year}{1993}),
  \bibinfo{pages}{466--478}.
\newblock
\urldef\tempurl%
\url{https://doi.org/10.1109/18.212277}
\showDOI{\tempurl}


\bibitem[\protect\citeauthoryear{Varma and Maguluri}{Varma and
  Maguluri}{2020}]%
        {Varma2019}
\bibfield{author}{\bibinfo{person}{Sushil~Mahavir Varma} {and}
  \bibinfo{person}{Siva~Theja Maguluri}.} \bibinfo{year}{2020}\natexlab{}.
\newblock \showarticletitle{Throughput Optimal Routing in Blockchain Based
  Payment Systems}.
\newblock \bibinfo{journal}{\emph{CoRR}}  \bibinfo{volume}{abs/2001.05299}
  (\bibinfo{year}{2020}).
\newblock
\showeprint[arxiv]{2001.05299}
\urldef\tempurl%
\url{https://arxiv.org/abs/2001.05299}
\showURL{%
\tempurl}


\bibitem[\protect\citeauthoryear{Wang, Xu, Jin, and Wang}{Wang
  et~al\mbox{.}}{2019}]%
        {Wang2019}
\bibfield{author}{\bibinfo{person}{Peng Wang}, \bibinfo{person}{Hong Xu},
  \bibinfo{person}{Xin Jin}, {and} \bibinfo{person}{Tao Wang}.}
  \bibinfo{year}{2019}\natexlab{}.
\newblock \showarticletitle{Flash: Efficient Dynamic Routing for Offchain
  Networks}. In \bibinfo{booktitle}{\emph{Proceedings of the 15th International
  Conference on Emerging Networking Experiments And Technologies}} (Orlando,
  Florida) \emph{(\bibinfo{series}{CoNEXT '19})}.
  \bibinfo{publisher}{Association for Computing Machinery},
  \bibinfo{address}{New York, NY, USA}, \bibinfo{pages}{370–381}.
\newblock
\showISBNx{9781450369985}
\urldef\tempurl%
\url{https://doi.org/10.1145/3359989.3365411}
\showDOI{\tempurl}


\bibitem[\protect\citeauthoryear{{Yu}, {Xue}, {Kilari}, {Yang}, and
  {Tang}}{{Yu} et~al\mbox{.}}{2018}]%
        {Yu2018a}
\bibfield{author}{\bibinfo{person}{Ruozhou {Yu}}, \bibinfo{person}{Guoliang
  {Xue}}, \bibinfo{person}{Vishnu~Teja {Kilari}}, \bibinfo{person}{Dejun
  {Yang}}, {and} \bibinfo{person}{Jian {Tang}}.}
  \bibinfo{year}{2018}\natexlab{}.
\newblock \showarticletitle{{CoinExpress}: A Fast Payment Routing Mechanism in
  Blockchain-Based Payment Channel Networks}. In \bibinfo{booktitle}{\emph{2018
  27th International Conference on Computer Communication and Networks
  (ICCCN)}}. \bibinfo{pages}{1--9}.
\newblock
\urldef\tempurl%
\url{https://doi.org/10.1109/ICCCN.2018.8487351}
\showDOI{\tempurl}


\end{thebibliography}

\newpage
\appendix

\section{Summary of notation}
\label{app:summary-of-notation}

\begin{table}[h]
    \caption{Notation used throughout the paper}
    \label{table:notation}
    \centering
	\begin{tabular}{|c|l|}
		\hline 
% 		\toprule
        \textbf{Symbol}         & \textbf{Meaning}\\
		\hline 
% 		\midrule
		$A$, $B$                & Nodes of the channel\\
		\hline
		$C$                     & Capacity of the channel\\
		\hline
		$Q^A(t)$                & Balance of node $A$ on the channel at time $t$\\
		\hline
		$t_n^A$                 & Arrival time of $n$-th transaction of node $A$\\
		\hline		
		$v_n^A$                 & Value (amount) of $n$-th transaction of node $A$\\
		\hline		
		$d_n^A$                 & Maximum buffering time of $n$-th transaction of node $A$\\
		\hline
		$\tau_n^A$              & Deadline expiration time of $n$-th transaction of node $A$\\
        \hline
        $D_k^A(t)$              & Remaining time until expiration of $k$-th transaction in node $A$'s buffer at time $t$\\
		\hline		
        $v_k^A(t)$              & Value (amount) $k$-th transaction in node $A$'s buffer at time $t$\\
		\hline
		$K^A(t)$                & Number of pending transactions in node $A$'s buffer at time $t$\\
        \hline
		$T_{\text{arrival}}$    & Sequence of all transaction arrival times on both sides of the channel\\
		\hline
		$T_{\text{expiration}}$    & Sequence of all deadline expiration times on both sides of the channel\\
		\hline		
		$x(t)$                  & System state at time $t$ (channel balances and buffer contents)\\
		\hline
		$u(t)$                  & Action taken at time $t$\\
		\hline
		$U(x(t))$               & Action space at time $t$\\
		\hline
		$\Tilde{v}_{EX}^{u(t)}(t)$      & Total amount processed by the channel at time $t$\\
		\hline
        $\Tilde{v}_{DR}^{u(t)}(t)$	    & Total amount rejected by the channel at time $t$\\
        \hline
        $\pi$                   & Control policy\\
        \hline              
		$\Pi$                   & Set of admissible control policies\\
		\hline
		$V_{\text{total}}(t)$   & total amount of arrivals until time $t$\\
        \hline
		$S^{\pi}(t)$            & Total channel throughput up to time $t$ under policy $\pi$\\
		\hline
		$R^{\pi}(t)$            & Total channel blockage (rejected amount) up to time $t$ under policy $\pi$\\
		\hline
		$P^{\pi}(t)$            & Amount of pending transactions at time $t$ under policy $\pi$\\
% 		\hline
% 		$J^{\pi}$               & Expected long-term average channel throughput under policy $\pi$\\
		\hline 
		% \bottomrule				
	\end{tabular}
\end{table}

\section{Proof of Lemma \ref{lemma:channel-with-buffers}}
\label{app:proof-of-main-lemma}

We restate Lemma \ref{lemma:channel-with-buffers} here for the reader's convenience.

\begin{L2}
For every policy $\pi \in \Pi^{DE}$, there exists a policy $\Tilde{\pi} \in \Pi^{DE}$ that acts similarly to PMDE at $t = \tau_1$ and is such that when the system is in state $x(0)$ at $t=0$ and policies $\pi$ and $\Tilde{\pi}$ act on it, the corresponding total rejected amount processes $R$ and $\Tilde{R}$ can be constructed via an appropriate coupling of the arrival processes so that
\begin{equation}
% \label{eqn:lemma-with-buffers-main-blockage}
\Tilde{R}(t) \leq R(t), t \in \tau_1, \tau_2, \dots
\end{equation}
\end{L2}

\begin{proof}
We construct $\Tilde{\pi}$ and couple the blockage processes under $\pi$ and $\Tilde{\pi}$ so that \eqref{eqn:lemma-with-buffers-main-blockage} holds.
Let the first transaction arrivals be identical (same arrival times, values and deadlines) under both policies.
Denote the time instant when the first deadline expiration occurs by $\tau_1$.
%, and the values of the $n$ in number transactions expiring at $\tau_1$ by $v_1^1, \dots, v_1^n$.
Without loss of generality, let $p_1^A$ be from node $A$ to node $B$ be one of the transactions expiring at $\tau_1$.
We distinguish the following cases based on the actions policy $\Tilde{\pi}$ might take at $\tau_1$:
\begin{enumerate}
    \item $\Tilde{\pi}$ drops $p_1^A$ at $\tau_1$. 
    The only reason why this would happen, since $\Tilde{\pi}$ mimics PMDE at $\tau_1$, is if $p_1^A$ is infeasible and there is no pending feasible transaction on the opposite side.
    
    The fact that $p_1^A$ is infeasible, since all transaction amounts are of the same fixed value, means that all transactions in the same direction are individually infeasible for $\pi$ as well.
    
    The fact that there is no pending feasible transaction on the opposite side means that $\pi$ cannot process any individual transaction in the opposite direction.
    
    Therefore, $\pi$ has no choice but to drop $p_1^A$, and possibly drop some other transactions.
    Denote the set of these other transactions dropped by $\pi$ at $\tau_1$ by $P_1^d$.
    
    At the next expiration time $\tau_2$, we let $\Tilde{\pi}$ operate in two phases:
    in the first phase, we let $\Tilde{\pi}$ drop all transactions in $P_1^d$.
    So now both the states and the blockages under both $\pi$ and $\Tilde{\pi}$ are the same.
    In the second phase at $\tau_2$, let $\Tilde{\pi}$ match the action of that $\pi$ takes at $\tau_2$.
    For all future expiration times, we let $\Tilde{\pi}$ be identical to $\pi$.
    Then the blockage processes under $\pi$ and $\Tilde{\pi}$ are identical, and therefore \eqref{eqn:lemma-with-buffers-main-blockage} holds.

    \item $p_1$ is individually feasible at $\tau_1$, and $\Tilde{\pi}$ processes it at $t=\tau_1$
    
    We distinguish cases based on what policy $\pi$ does at $\tau_1$.
    
    \begin{enumerate}
        \item At $\tau_1$, $\pi$ processes $p_1^A$, drops some transactions from possibly both sides (set $P_1^d$) and processes some other transactions from possibly both sides (set $P_1^p$).
        Then $R(\tau_1) = |P_1^d|v$ and $\Tilde{R}(\tau_1) = 0$.
        
        Let $\tau_2$ be the next time of deadline expiration. 
        At $\tau_2$, we let $\Tilde{\pi}$ operate in two phases.
        In the first phase, we let $\Tilde{\pi}$ drop all transactions in $P_1^d$ and process all transactions in $P_1^p$ at $\tau_2$, just like $\pi$ did at $\tau_1$.
        Now the states under $\pi$ and $\Tilde{\pi}$ are the same, and the same is true for the blockages: $\Tilde{R} = R = |P_1^d|v$.
        In the second phase, we let $\Tilde{\pi}$ be identical to $\pi$ at $\tau_2$.
        For all future expiration times, we also let $\Tilde{\pi}$ be identical to $\pi$, and both the states and the blockages under $\pi$ and $\Tilde{\pi}$ match.
        
        \item At $\tau_1$, $\pi$ drops $p_1^A$, drops some transactions (set $P_1^d$) and processes some other transactions (set $P_1^p$).
        Then $R(\tau_1) = (|P_1^d|+1)v$ and $\Tilde{R}(\tau_1) = 0$.
        
        Let $\tau_2$ be the next time of deadline expiration. 
        At $\tau_2$, we let $\Tilde{\pi}$ operate in two phases.
        In the first phase, we let $\Tilde{\pi}$ drop all transactions in $P_1^d$ and attempt to process all transactions in $P_1^p$ at $\tau_2$, just like $\pi$ did at $\tau_1$.
        Depending on whether the latter is possible, we distinguish the following cases:
        
        \begin{enumerate}
            \item If this is possible, then now the states under $\pi$ and $\Tilde{\pi}$ are almost the same, with the only difference being that node $A$ has processed one more transaction from A to B under $\Tilde{\pi}$.
            So, at that moment, for the balances under the two policies we have $\left( \Tilde{Q}^A, \Tilde{Q}^B \right) = (Q^A - v, Q^B + v)$, and for the blockages we have $\Tilde{R} = R + v$.
            In the second phase at $\tau_2$, and at subsequent deadline expiration times, we let $\Tilde{\pi}$ match $\pi$ (and thus the relationships between the balances and the blockages under the two policies remain the same).
            This will be always possible except if at some point $\pi$ executes some transaction from A to B and since A's balance under $\Tilde{\pi}$ is less than under $\pi$, $\Tilde{\pi}$ is not able to process it.
            At that moment, we let $\Tilde{\pi}$ drop that infeasible transaction and match $\pi$ in the rest of $\pi$'s actions.
            Then we have $\left( \Tilde{Q}^A, \Tilde{Q}^B \right) = (Q^A, Q^B)$ and $\Tilde{R} = R$.
            For all future expiration times, we also let $\Tilde{\pi}$ be identical to $\pi$, and both the states and the blockages under $\pi$ and $\Tilde{\pi}$ match.
            
            \item If this is not possible, the only transaction feasible under $\pi$ but not under $\Tilde{\pi}$ must be from A to B.
            We let $\Tilde{\pi}$ drop that transaction and follow $\pi$ in all other transactions $\pi$ processes or drops.
            So now, $\left( \Tilde{Q}^A, \Tilde{Q}^B \right) = (Q^A, Q^B)$ and $\Tilde{R} = R$.
            For all future expiration times, we let $\Tilde{\pi}$ be identical to $\pi$, and both the states and the blockages under $\pi$ and $\Tilde{\pi}$ match.
        \end{enumerate}
    \end{enumerate}

    \item $p_1^A$ is individually infeasible at $\tau_1$, but $\Tilde{\pi}$ processes $p_1^A$ at $t=\tau_1$ by matching it with a transaction $\Tilde{p}_1^B$ from B to A
    
    We distinguish cases based on what policy $\pi$ does at $\tau_1$.
    
    \begin{enumerate}
        \item At $\tau_1$, $\pi$ processes $p_1^A$, drops some transactions from possibly both sides (set $P_1^d$) and processes some other transactions from possibly both sides (set $P_1^p$).
        Then $R(\tau_1) = |P_1^d|v$ and $\Tilde{R}(\tau_1) = 0$.
        Since $p_1^A$ is individually infeasible at $\tau_1$ (for both policies), the only way $\pi$ can process $p_1^A$ at $\tau_1$ is if it matches it with another transaction from the opposite direction; call the matched transaction $p_1^B \in P_1^p$.
        
        Let $\tau_2$ be the next time of deadline expiration. 
        At $\tau_2$, we let $\Tilde{\pi}$ operate in two phases.
        We distinguish cases based on what policy $\pi$ does with transaction $\Tilde{p}_1^B$ at $\tau_1$.
        
        \begin{enumerate}
            \item $\Tilde{p}_1^B \in P_1^p$ (i.e. $\Tilde{p}_1^B$ is processed by $\pi$ at $\tau_1$)
            
            % We distinguish the following cases:
            % \begin{enumerate}
            %     \item $p_1^B \equiv \Tilde{p}_1^B$
            
                In this case, in the first phase at $\tau_2$ we let $\Tilde{\pi}$ drop all transactions in $P_1^d \setminus \{p_1^B\}$ and process all transactions in $P_1^p$ at $\tau_2$, just like $\pi$ did at $\tau_1$.
                Now the states under $\pi$ and $\Tilde{\pi}$ are the same, and the same is true for the blockages: $\Tilde{R} = R = |P_1^d|v$.
                In the second phase, we let $\Tilde{\pi}$ be identical to $\pi$ at $\tau_2$.
                For all future expiration times, we also let $\Tilde{\pi}$ be identical to $\pi$, and both the states and the blockages under $\pi$ and $\Tilde{\pi}$ match.
        
            %     \item $p_1^B \not\equiv \Tilde{p}_1^B$
                
            % \end{enumerate}
        
            \item $\Tilde{p}_1^B \notin P_1^p$ and $\Tilde{p}_1^B \in P_1^d$ (i.e. $\Tilde{p}_1^B$ is dropped by $\pi$ at $\tau_1$, so it is not in B's buffer anymore under $\pi$ at $\tau_2$)
            
            In this case, in the first phase at $\tau_2$ we let $\Tilde{\pi}$ drop all transactions in $P_1^d$, drop also $p_1^B$ from $P_1^p$, and process all transactions in $P_1^p \setminus \{p_1^B\}$ at $\tau_2$.
            Now the states under $\pi$ and $\Tilde{\pi}$ are the same, and the same is true for the blockages: $\Tilde{R} = R = (|P_1^d|+1)v$.
            In the second phase, we let $\Tilde{\pi}$ be identical to $\pi$ at $\tau_2$.
            For all future expiration times, we also let $\Tilde{\pi}$ be identical to $\pi$, and both the states and the blockages under $\pi$ and $\Tilde{\pi}$ match.
            
            \item $\Tilde{p}_1^B \notin P_1^p$ and $\Tilde{p}_1^B \notin P_1^d$ (i.e. $\Tilde{p}_1^B$ is neither processed nor dropped by $\pi$ at $\tau_1$, so it is still in B's buffer under $\pi$ at $\tau_2$)
            
            In this case, in the first phase at $\tau_2$ we let $\Tilde{\pi}$ drop all transactions in $P_1^d$ and attempt to process all transactions in $P_1^p$ at $\tau_2$, just like $\pi$ did at $\tau_1$.
            
            Depending on whether the latter is possible, we distinguish the following cases:
            
            \begin{enumerate}
                \item If this is possible, then now the states under $\pi$ and $\Tilde{\pi}$ are almost the same, with the only difference being that node $A$ has processed one more transaction ($\Tilde{p}_1^B$) from B to A under $\Tilde{\pi}$.
                So, at that moment, we have $\left( \Tilde{Q}^A, \Tilde{Q}^B \right) = (Q^A + v, Q^B - v)$ and $\Tilde{R} = R - v$.
                In the second phase at $\tau_2$, and at subsequent deadline expiration times, we let $\Tilde{\pi}$ match $\pi$ (and thus the relationships between the balances and the blockages under the two policies remain the same).
                This will be always possible except if at some point $\pi$ executes some transaction from B to A and since B's balance under $\Tilde{\pi}$ is less than under $\pi$, $\Tilde{\pi}$ is not able to process it.
                At that moment, we let $\Tilde{\pi}$ drop that infeasible transaction and match $\pi$ in the rest of $\pi$'s actions.
                Then we have $\left( \Tilde{Q}^A, \Tilde{Q}^B \right) = (Q^A, Q^B)$ and $\Tilde{R} = R$.
                For all future expiration times, we let $\Tilde{\pi}$ be identical to $\pi$, and both the states and the blockages under $\pi$ and $\Tilde{\pi}$ match.
                
                \item If this is not possible, the only transaction feasible under $\pi$ but not under $\Tilde{\pi}$ must be from B to A (so in the same direction as $p_1^B$, or $p_1^B$ itself).
                We let $\Tilde{\pi}$ drop $p_1^B$.
                So now, $\left( \Tilde{Q}^A, \Tilde{Q}^B \right) = (Q^A, Q^B)$ and $\Tilde{R} = R$.
                
                In the second phase, we let $\Tilde{\pi}$ follow $\pi$ in all other transactions $\pi$ processes or drops at $\tau_2$.
                For all future expiration times, we let $\Tilde{\pi}$ be identical to $\pi$, and both the states and the blockages under $\pi$ and $\Tilde{\pi}$ match.
            \end{enumerate}
        \end{enumerate}       
        
        \item At $\tau_1$, $\pi$ drops $p_1^A$, drops some transactions from possibly both sides (set $P_1^d$) and processes some other transactions from possibly both sides (set $P_1^p$).
        Then $R(\tau_1) = (|P_1^d|+1)v$ and $\Tilde{R}(\tau_1) = 0$.
        
        Let $\tau_2$ be the next time of deadline expiration. 
        At $\tau_2$, we let $\Tilde{\pi}$ operate in two phases.
        We distinguish cases based on what policy $\pi$ does with transaction $\Tilde{p}_1^B$ at $\tau_1$.
        
        \begin{enumerate}
            \item $\Tilde{p}_1^B \in P_1^p$ (i.e. $\Tilde{p}_1^B$ is processed by $\pi$ at $\tau_1$)
        
            In this case, in the first phase at $\tau_2$ we let $\Tilde{\pi}$ drop all transactions in $P_1^d$ and process all transactions in $P_1^p \setminus \{\Tilde{p}_1^B\}$, just like $\pi$ did at $\tau_1$.
            Depending on whether the latter is possible, we distinguish the following cases:
        
                % Now the states under $\pi$ and $\Tilde{\pi}$ are the same, and the same is true for the blockages: $\Tilde{R} = R = |P_1^d|v$.
                % In the second phase, we let $\Tilde{\pi}$ be identical to $\pi$ at $\tau_2$.
                % For all future expiration times, we also let $\Tilde{\pi}$ be identical to $\pi$, and both the states and the blockages under $\pi$ and $\Tilde{\pi}$ match.
        
                \begin{enumerate}
                    \item If this is possible, then now the states under $\pi$ and $\Tilde{\pi}$ are almost the same, with the only difference being that node $A$ has processed one more transaction ($p_1^A$) from A to B under $\Tilde{\pi}$.
                    So, at that moment, we have $\left( \Tilde{Q}^A, \Tilde{Q}^B \right) = (Q^A - v, Q^B + v)$ and $\Tilde{R} = R + v$.
                    In the second phase at $\tau_2$, and at subsequent deadline expiration times, we let $\Tilde{\pi}$ match $\pi$ (and thus the relationships between the balances and the blockages under the two policies remain the same).
                    This will be always possible except if at some point $\pi$ executes some transaction from A to B and since A's balance under $\Tilde{\pi}$ is less than under $\pi$, $\Tilde{\pi}$ is not able to process it.
                    At that moment, we let $\Tilde{\pi}$ drop that infeasible transaction and match $\pi$ in the rest of $\pi$'s actions.
                    Then we have $\left( \Tilde{Q}^A, \Tilde{Q}^B \right) = (Q^A, Q^B)$ and $\Tilde{R} = R$.
                    For all future expiration times, we let $\Tilde{\pi}$ be identical to $\pi$, and both the states and the blockages under $\pi$ and $\Tilde{\pi}$ match.
                    
                    \item If this is not possible, the only transaction feasible under $\pi$ but not under $\Tilde{\pi}$ must be from A to B (so in the same direction as $p_1^A$, or $p_1^A$ itself).
                    We let $\Tilde{\pi}$ drop that transaction.
                    So now, $\left( \Tilde{Q}^A, \Tilde{Q}^B \right) = (Q^A, Q^B)$ and $\Tilde{R} = R$.
                    
                    In the second phase, we let $\Tilde{\pi}$ follow $\pi$ in all other transactions $\pi$ processes or drops at $\tau_2$.
                    For all future expiration times, we let $\Tilde{\pi}$ be identical to $\pi$, and both the states and the blockages under $\pi$ and $\Tilde{\pi}$ match.
                \end{enumerate}
        
            \item $\Tilde{p}_1^B \notin P_1^p$ and $\Tilde{p}_1^B \in P_1^d$ (i.e. $\Tilde{p}_1^B$ is dropped by $\pi$ at $\tau_1$, so it is not in B's buffer anymore under $\pi$ at $\tau_2$)
            
            In this case, in the first phase at $\tau_2$ we let $\Tilde{\pi}$ drop all transactions in $P_1^d$, and process all transactions in $P_1^p$.
            Now the states (balances and buffer contents) under $\pi$ and $\Tilde{\pi}$ are the same.
            For the blockages, we have: $\Tilde{R} = R - 2v \leq R$.
            In the second phase, we let $\Tilde{\pi}$ be identical to $\pi$ at $\tau_2$.
            For all future expiration times, we also let $\Tilde{\pi}$ be identical to $\pi$, and both the states and the blockages under $\pi$ and $\Tilde{\pi}$ match.
            So \eqref{eqn:lemma-with-buffers-main-blockage} holds for all expiration times.
            
            \item $\Tilde{p}_1^B \notin P_1^p$ and $\Tilde{p}_1^B \notin P_1^d$ (i.e. $\Tilde{p}_1^B$ is neither processed nor dropped by $\pi$ at $\tau_1$, so it is still in B's buffer under $\pi$ at $\tau_2$)
            
            In this case, in the first phase at $\tau_2$ we let $\Tilde{\pi}$ drop all transactions in $P_1^d$ and attempt to process all transactions in $P_1^p$ at $\tau_2$, just like $\pi$ did at $\tau_1$.
            This will be always possible, as the only difference in the states under $\pi$ and $\Tilde{\pi}$ is that B's buffer contains $\Tilde{p}_1^B$ under $\pi$ but not under $\Tilde{\pi}$.
            In terms of executed transactions, $\Tilde{\pi}$ compared to $\pi$ has additionally processed the pair $\left( p_1^A, \Tilde{p}_1^B \right)$ and pairs have no net effect on the balances. So $\left( \Tilde{Q}^A, \Tilde{Q}^B \right) = (Q^A, Q^B)$ and $\Tilde{R} = R - v \leq R$.
            
            In the second phase at $\tau_2$, and at subsequent deadline expiration times, we let $\Tilde{\pi}$ match $\pi$ (and thus the relationships between the balances and the blockages under the two policies remain the same).
            This will be always possible except if at some point $\pi$ decides to execute or drop $\Tilde{p}_1^B$ that is still in B's buffer under $\pi$ but not under $\Tilde{\pi}$.
            
            \begin{enumerate}
                \item If $\pi$ drops $\Tilde{p}_1^B$, then the states under $\pi$ and $\Tilde{\pi}$ completely match (balances and buffer contents), and we have $\Tilde{R} = R - v \leq R$.
                At that moment, and for all future expiration times, we let $\Tilde{\pi}$ be identical to $\pi$, and thus both the states and the blockages under $\pi$ and $\Tilde{\pi}$ match.
                
                \item If $\pi$ processes $\Tilde{p}_1^B$, then $\Tilde{\pi}$ cannot do the same, and thus $\pi$ has processed one more transaction from B to A than $\Tilde{\pi}$.
                So we have $\left( \Tilde{Q}^A, \Tilde{Q}^B \right) = (Q^A - v, Q^B + v)$, same buffer contents, and $\Tilde{R} = R - v \leq R$.
                
                From then on, we let $\Tilde{\pi}$ match $\pi$ for as long as this is possible (and thus the relationships between the balances and the blockages under the two policies remain the same).
                The only reason why $\Tilde{\pi}$ at some point might not be able to match $\pi$ is if it $\pi$ executes a transaction from B to A that $\Tilde{\pi}$ cannot because B's balance under $\Tilde{\pi}$ is less than under $\pi$.
                At that time, we let $\Tilde{\pi}$ drop that transaction and match $\pi$ in all its other actions.
                So now $\left( \Tilde{Q}^A, \Tilde{Q}^B \right) = (Q^A, Q^B)$, the buffer contents are the same, and $\Tilde{R} = R$.
                For all future expiration times, we let $\Tilde{\pi}$ be identical to $\pi$, and both the states and the blockages under $\pi$ and $\Tilde{\pi}$ match.

            \end{enumerate}
        \end{enumerate}
    \end{enumerate}
\end{enumerate}
Thus, in all possible cases, it is possible to couple $\Tilde{\pi}$ with $\pi$ so that \eqref{lemma:channel-with-buffers} is satisfied.
This concludes the proof of the lemma.
\end{proof}

\section{Proof of Theorem \ref{thm:analytical-calculation-of-success-rate}}
\label{app:proof-of-analytical-calculation}

We restate Theorem \ref{thm:analytical-calculation-of-success-rate} here for the reader's convenience.

\begin{T2}
% \begin{oneshot}{\ref{thm:analytical-calculation-of-success-rate}}
% \begin{newreptheorem}{thm:analytical-calculation-of-success-rate}{Theorem}
For a single channel between nodes $A$ and $B$ with capacity $C$, and Poisson transaction arrivals with rates $\lambda_A \neq \lambda_B$ and fixed amounts equal to $v$, the maximum possible success rate of the channel is
\begin{equation}
SR_{\text{opt}} = \lambda_A \left( 1 - \frac{\lambda_B/\lambda_A - 1}{(\lambda_B/\lambda_A)^{\Tilde{C}+1} - 1} \right) 
+ \lambda_B \left( 1 - \left(\frac{\lambda_B}{\lambda_A}\right)^{\Tilde{C}} \frac{\lambda_B/\lambda_A - 1}{(\lambda_B/\lambda_A)^{\Tilde{C}+1} - 1} \right)
\end{equation}
where $\Tilde{C} = \lfloor \frac{C}{v} \rfloor$.\\

When $\lambda_A = \lambda_B = \lambda$, the maximum possible success rate is
\begin{equation}
SR_{\text{opt}} = \frac{2\lambda \Tilde{C}}{\Tilde{C}+1}
\end{equation}
% \end{oneshot}
% \end{newreptheorem}
\end{T2}

\begin{proof}
The maximum possible success rate of the channel is the success rate under the optimal policy PFI.

We focus on the balance of node $A$, which has an initial value of $b_A$. 
Over time, it forms a continuous-time, time-homogeneous Markov chain that is a birth-death process with states $\{0, 1, \dots, C\}$.
Since all transactions are of amount $v$, only the states that are a multiple of $v$ away from $b_A$ are reachable.
Therefore, we reduce this Markov chain to another one with fewer states: $\{0, 1, \dots, \Tilde{C}\}$, where $\Tilde{C} = \lfloor C/v \rfloor$ and state $k$ in the new Markov chain corresponds to state $\bmod(b_A, v) + k v$ in the initial Markov chain, $k = 1, \dots, \Tilde{C} - 1$.

The state transition diagram of the new Markov chain is the following:

\begin{figure}[h]
    \centering
    \includegraphics[scale=0.5]{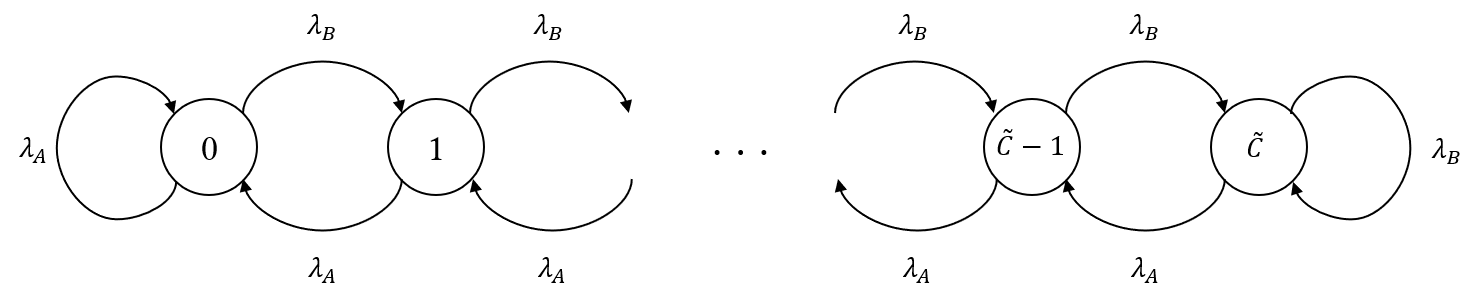}
    \label{fig:state_transition_diagram}
\end{figure}

Let $\pi = (\pi_1, \dots, \pi_{\Tilde{C}})$ be the stationary distribution. 

The long-term rejection rate $RR_{\text{opt}} $ (fraction of rejected transactions) and success rate $SR_{\text{opt}}$ of the channel can be calculated as follows:

\begin{equation}
RR_{\text{opt}}  = \lambda_A \pi_0 + \lambda_B \pi_C
\end{equation}

\begin{align}
    SR_{\text{opt}}   &= \lambda_A + \lambda_B - RR_{\text{opt}}  \\
        &= \lambda_A (1 - \pi_0) + \lambda_B (1 - \pi_{\Tilde{C}})    
\end{align}

Therefore, we need to calculate the stationary distribution.
The local balance equations are:
\begin{equation}
\lambda_B \pi_k = \lambda_A \pi_{k+1}
\end{equation}

So $\pi_{k+1} = \frac{\lambda_B}{\lambda_A} \pi_k = \left(\frac{\lambda_B}{\lambda_A}\right)^k \pi_0,~k = 0, \dots, \Tilde{C} - 1$.

The normalization constraint hence yields

\begin{equation}
\sum_{k=0}^{\Tilde{C}} \pi_k = 1
\implies \pi_0 \sum_{k=0}^{\Tilde{C}} \left(\frac{\lambda_B}{\lambda_A}\right)^k = 1
\end{equation}

% The local balance equations together with the normalization constraint for the stationary probabilities yield:
% \begin{equation}
% \pi_0 \sum_{k=0}^{\Tilde{C}} \left(\frac{\lambda_B}{\lambda_A}\right)^k = 1
% \end{equation}

We now distinguish between the two cases:\\
If $\lambda_A \neq \lambda_B$, then:

\begin{equation}
\pi_0 = \frac{\lambda_B/\lambda_A - 1}{(\lambda_B/\lambda_A)^{\Tilde{C}+1} - 1}
\end{equation}
and 
\begin{equation}
\pi_k = \left(\frac{\lambda_B}{\lambda_A}\right)^k \pi_0,~k = 1, \dots, \Tilde{C}
\end{equation}

If $\lambda_A = \lambda_B = \lambda$, then
\begin{equation}
\pi_k = \frac{1}{\Tilde{C}+1},~k = 0, 1, \dots, \Tilde{C}
\end{equation}

Plugging the stationary distribution into the success rate formula completes the proof.
\end{proof}

\end{document}